\documentclass[11pt,letterpaper]{article}
\usepackage{mathtools}
\usepackage[english]{babel}
\usepackage[T1]{fontenc}
\usepackage[utf8]{inputenc}
\usepackage{lmodern}
\usepackage{ulem} 
\usepackage{amsmath}
\usepackage{amssymb}
\usepackage{amsthm}
\usepackage{amsfonts}
\usepackage{mathdots}
\usepackage{braket}
\usepackage{enumitem}
\usepackage{graphicx}
\usepackage{xcolor}
\usepackage{url}
\usepackage{soul}
\usepackage{hyperref}
\usepackage{arydshln}
\usepackage[top=2cm, bottom=2cm, left=2cm, right=2cm, columnsep=20pt]{geometry}
\usepackage{tikz}
\usetikzlibrary{calc}
\newcommand{\blk}[2]{\tikz[remember picture,baseline=(#1.base)]{\node[inner sep=1pt,outer sep=0pt] (#1) {$#2$};}}

\newtheorem{proposition}{Proposition}
\newtheorem{theorem}{Theorem}
\newtheorem{remark}{Remark}

\begin{document}

\title{\bf Para families of orthogonal polynomials}

\author{
St\'ephane Z.\ Beaulac\textsuperscript{$1$}\footnote{E-mail: stephane.jr.beaulac@umontreal.ca}~,
Nicolas Cramp\'e\textsuperscript{$1,2$}\footnote{E-mail: nicolas.crampe@cnrs.fr}~,
Quentin Labriet\textsuperscript{$1$}\footnote{E-mail: quentin.labriet@umontreal.ca}~,\\
Lucia Morey\textsuperscript{$1$}\footnote{E-mail: lucia.morey@umontreal.ca}~,
Rafael I.\ Nepomechie\textsuperscript{$3$}\footnote{E-mail: nepomechie@miami.edu}~,
Luc Vinet\textsuperscript{$1$}\footnote{E-mail: luc.vinet@umontreal.ca}~.
\vspace{0.5cm}\\
\textsuperscript{$1$}
\small Centre de Recherches Math\'ematiques, Universit\'e de Montr\'eal, P.O.\ Box 6128,\\
\small Centre-ville Station, Montr\'eal (Qu\'ebec), H3C 3J7, Canada.\vspace{0.2cm}\\
\textsuperscript{$2$}
\small Laboratoire d'Annecy de Physique Th\'eorique, 9 Chemin de Bellevue,\\
\small BP 110, Annecy-le-Vieux, F-74941 Annecy Cedex, France.\vspace{0.2cm}\\
\textsuperscript{$3$}
\small Department of Physics, University of Miami, P.O.\ Box 248046,\\
\small Coral Gables, FL 33124, USA.\vspace{0.2cm}\\
}
\date{}
\maketitle

\begin{center}
\begin{minipage}{13cm}
\begin{center}
{\bf Abstract}\\
\end{center}
\small
The para-Krawtchouk polynomials arose in the search for spin chains with
perfect state transfer; the para-Racah, $q$-para-Racah and para-Bannai--Ito
polynomials followed, obtained through singular truncations of families
of the Askey scheme and used in turn to design spin chains.  All are
orthogonal on bi-lattices; the prefix goes back to the para-Krawtchouk case,
whose spectrum is that of the parabose oscillator.  The representations that
these polynomials provide of the corresponding algebras --- Hahn, Racah,
$q$-Racah, Bannai--Ito and complementary Bannai--Ito --- are shown to be
reducible.  Each decomposes into a direct sum of two irreducible modules of
the same algebra, one attached to each of the two interlaced sub-lattices,
and the parameters of the submodules are given.  The para polynomials are
thereby expressed in terms of the classical polynomials that the submodules
carry, and the pairing $\lambda_n=\lambda_{N-n}$ of their eigenvalues is
accounted for; the three features of the bi-lattice setting on which the
reduction rests are identified.  The properties of the para families are then
presented in a unified way in the light of this decomposition, together with
the constructions through which these families arise and the settings in
which they have been put to  use.
\end{minipage}
\end{center}

\vspace{0.5cm}

\section{Introduction}
\label{intro}

The transport of a quantum state from one end of a chain of coupled spins to
the other can sometimes be effected by the dynamics of the chain itself, without external
control.  In the XX model with inhomogeneous nearest-neighbour couplings, the
one-excitation dynamics is governed by a Jacobi matrix $J$, and the transfer
is perfect when the differences between neighbouring eigenvalues of $J$ are
odd integers up to a common factor \cite{kay,PST_construct}.  Designing such a
chain is therefore an inverse spectral problem: one posits a spectrum meeting
that condition and reconstructs the couplings.  Since a Jacobi matrix carries
a three-term recurrence relation, the reconstruction is performed by the
orthogonal polynomials attached to it, and the choice of spectrum is the
choice of a grid.

The simplest grid is uniform.  It yields the binomial distribution, the
Krawtchouk polynomials, and the model of Albanese, Christandl, Datta and
Ekert \cite{ACDE}.  Another natural case is a bi-lattice, that is, the union of
two uniform lattices shifted with respect to one another. Orthogonal
polynomials on such grids were considered in \cite{SVA}; the grid
that meets the transfer condition produced a family of orthogonal polynomials whose recurrence
coefficients are explicit, and which were identified in
\cite{parakrawtchouk}. It was called para-Krawtchouk, the prefix recording
that its spectrum is that of the parabose oscillator \cite{rosenblum}.%
\footnote{The prefix is unrelated to the para-orthogonal polynomials on the
unit circle of \cite{JNT}, which are the combinations
$\Phi_n(z)+\tau\Phi_n^{*}(z)$, $|\tau|=1$, used in Szeg\H{o} quadrature.}

Three further families followed, obtained the other way round --- first as
polynomials and were then used to construct chains.  Each arises from a truncation of a family of the
Askey scheme that terminates the recurrence relation without terminating the
hypergeometric series in the usual manner: the para-Racah polynomials from the
Wilson polynomials \cite{pararacah}, the $q$-para-Racah polynomials from the
Askey--Wilson polynomials \cite{qPR}, and the para-Bannai--Ito polynomials
from the Bannai--Ito and complementary Bannai--Ito polynomials \cite{paraBI}.
All are orthogonal on bi-lattices, quadratic or $q$-quadratic in place of
linear.  These four families are what we shall call the para families. They have since occurred in two kinds of settings.
\paragraph{A. Quantum spin chains}
Beyond perfect state transfer, the para-Krawtchouk polynomials proved
revealing in the study of fractional revival, where the excitation is split
coherently between the two ends of the chain and later reconstituted.
Couplings realizing this effect had been obtained numerically \cite{BCB}; the
recurrence coefficients of the para-Krawtchouk polynomials supplied an
analytic model, with two parameters and a revival time that does not depend
on the length of the chain \cite{revival}.  Fractional revival on the
para-Krawtchouk chain has in turn been put to use: combined with dual-rail
encoding, it yields a protocol transferring a state perfectly in an average
time below the bound that holds for a single chain \cite{XKT}.  The other para
families supply
chains of the same kind on their respective grids, and the same circle of
ideas bears on almost perfect state transfer \cite{APST}.

\paragraph{B. Quadratic algebras}
Consider the difference operators $S=A_1(x)T_++A_2(x)T_-$ on the linear grid, with
$T_{\pm}f(x)=f(x\pm1)$, that raise the degree of a polynomial by at most one:
\begin{equation}
SP_n(x)=\widetilde{P}_{n+1}(x),
\qquad n=0,1,2,\ldots
\label{SHeundef}
\end{equation}
These are the S-Heun operators \cite{sklyanin}, second-order difference
operators without diagonal term subject to a degree-raising condition.  They
span a five-dimensional space, and a basis of it consists of one operator
that lowers the degree, two that leave it unchanged and two that raise it.
The interest of this small set is that everything usually attached to a
family of orthogonal polynomials is built from it by taking products of two operators:
the difference equation and the multiplication by the variable, that is the
two bispectral operators, as well as the forward, backward and contiguity
operators.  The para-Krawtchouk polynomials arise here in a way that will be
made precise in Subsection \ref{skl}: the S-Heun operators realize a
degenerate Sklyanin algebra $Skl_4$, that algebra has a finite-dimensional
representation exactly when a certain truncation condition holds, and the
para-Krawtchouk polynomials are the basis of that representation.  The
condition in question is the one of Section \ref{generalprop}, seen from the
algebraic side.  The para-Krawtchouk polynomials appear again among the
tridiagonal representations of the deformed Jordan plane \cite{jordan}.  More
generally, each para family provides a finite-dimensional representation of
the algebra of the family it is truncated from --- Hahn \cite{mutual}, Racah,
$q$-Racah, Bannai--Ito \cite{BI} and complementary Bannai--Ito \cite{CBI}. 

These representations are, however, reducible: they decompose naturally according to the two sub-lattices of the underlying grid. The difference operator that
each para family diagonalizes shifts within one sub-lattice and never across,
and its coefficients vanish at the two ends of each sub-lattice.  Reordering
the grid by parity therefore splits the operator into two blocks, while the
multiplication operator, being diagonal, preserves the splitting.  The module
carried by a para family is thus reducible, and we shall show that it is the
direct sum of two irreducible modules of the same algebra, whose parameters we
identify.  The para polynomials are then expressed in terms of the classical
polynomials that the two submodules carry, and the pairing
$\lambda_n=\lambda_{N-n}$ of the eigenvalues, which the para families have
exhibited from the outset, is accounted for.  That non-trivial sub-lattices
arise in this fashion is not without precedent: periodic reductions of the
factorization chain were seen to produce the Hahn algebra together with
sub-lattices of the spectral points \cite{SVZ}.  The present decomposition
gives a common ground on which the properties of the four families can be set
out together, and we take the occasion to bring together the work in which
they have appeared, which is dispersed over a number of papers.

The paper is organized as follows.  Section \ref{generalprop} obtains the
para-Krawtchouk polynomials from the continuous Hahn polynomials by a
non standard truncation and collects their properties.  Section
\ref{Hahn algebra} recalls the Hahn algebra they realize.  Section
\ref{reducibility} establishes the decomposition, in the spectral and in the
recurrence representation.  Section \ref{GM}
isolates the three features of the bi-lattice setting on which the
decomposition rests and carries it out for the para-Racah, $q$-para-Racah and
para-Bannai--Ito polynomials.  Section \ref{BIreal} treats the realization of
the para-Krawtchouk polynomials inside the complementary Bannai--Ito and
Bannai--Ito families, whose modules are irreducible for generic values of the
parameters: the values at which the para-Krawtchouk polynomials are recovered
meet a truncation condition, and the module ceases to be irreducible there.
Section
\ref{occurrences} sets out the constructions through which the para families
arise and the settings in which they have occurred.  Concluding remarks and
open problems follow.

\paragraph{Notations}
We recall the definitions of the generalized ($q$-)hypergeometric functions
used in this paper.  The generalized hypergeometric functions \cite{koekoek}
are defined by
\begin{align}
{}_{r+1}F_r\left({{-n,\;a_1,\; a_2,\; \cdots,\; a_{r}  }\atop
{b_1,\; b_2,\; \cdots,\; b_{r} }}\;\Bigg\vert \; z\right)=\sum_{k=0}^{n}
\frac{(-n,a_1,a_2,\cdots,a_r)_k}{k!(b_1,b_2,\cdots,b_r)_k}z^k\,,
\end{align}
for $r,n$ non-negative integers, where the Pochhammer symbols are
\begin{align}
(b_1,b_2,\cdots,b_r)_k=(b_1)_k(b_2)_k\cdots (b_r)_k\,,\qquad
(b)_k=b(b+1)\cdots (b+k-1)\,.
\end{align}
The $q$-hypergeometric functions \cite{koekoek} are defined by
\begin{equation}
{}_{r+1}\phi_r
\left(
\begin{matrix}
q^{-n}, a_1, a_2, \ldots, a_r\\
b_1, b_2, \ldots, b_r
\end{matrix}
\,;\, q, z
\right)
=
\sum_{k=0}^{n}
\frac{(q^{-n},a_1,a_2,\ldots,a_r;q)_k}
{(q,b_1,b_2,\ldots,b_r;q)_k}
z^k,
\label{qhypgeom}
\end{equation}
for $r,n$ non-negative integers, where the $q$-Pochhammer symbols are
\begin{equation}
(a_1,a_2,\ldots,a_r;q)_k
=
(a_1;q)_k(a_2;q)_k\cdots(a_r;q)_k,
\qquad
(a;q)_k
=
\prod_{j=0}^{k-1}(1-aq^j).
\label{qpoch}
\end{equation}

\section{The para-Krawtchouk polynomials}
\label{generalprop}

Finite families of orthogonal polynomials can be obtained from infinite ones
by imposing a termination condition on the three-term recurrence relation.
The para-Krawtchouk polynomials arise in this way from the continuous Hahn
polynomials \cite{sklyanin} through a non-standard truncation.  We
review the procedure and collect the properties of the polynomials it
produces.  The continuous Hahn polynomials, which depend on four parameters $a,b,c,d$, are recalled in Appendix
\ref{app:chahn}.

\paragraph{Truncation condition.}
The continuous Hahn polynomials $p_n(x;a,b,c,d)$ defined in \eqref{hyperc_hahn}, have recurrence coefficients $A_n, C_n$ given by \eqref{AC}. The recurrence relation \eqref{eq:3t-contHahn} terminates after a finite number of steps when
\begin{equation}
A_N C_{N+1}=0,
\end{equation}
Besides the standard truncation that leads to the Hahn polynomials ($a+d=-N$, with $N\in\mathbb{Z}_{>0}$), this can be achieved through the parametrization
\begin{equation}
1-(a+b+c+d)=N\in\mathbb{Z}_{>0}.
\label{truncation}
\end{equation}

\paragraph{Parametrization.}
Write $N=2j+p$, where $p=0,1$ according as $N$ is even or odd, respectively,
and $j$ is a non-negative integer ($j\geq 1$ when $p=0$).
Imposing \eqref{truncation} directly makes both the numerator and the
denominator of the recurrence coefficients vanish, so that $A_n$ and $C_n$
become indeterminate. To resolve this, we approach the truncation through the limit
\begin{equation}
c=-a-j+e_1t , \qquad b=-d-j+1-p+e_1t\,,\qquad t\to0\,,
\label{parametrization}
\end{equation}
where $e_1\neq 0$ is an arbitrary constant.

\paragraph{Recurrence relation.}
Substituting \eqref{parametrization} into the recurrence coefficients
$A_n,C_n$ of the continuous Hahn polynomials \eqref{AC} and taking the limit $t\to0$, one obtains \cite{sklyanin}
\begin{equation}
\begin{aligned}
A_n
&\to
-\frac{1}{2}
\frac{(2j+p-n)(2j+2p-2-2n+\gamma)}
{2(2n-2j-2p+1)},
\\[1em]
C_n
&\to
-\frac{1}{2}
\frac{n(2j+2-2n-\gamma)}
{2(2n-2j-1)},
\end{aligned}
\end{equation}
where $
    \gamma=(b+c)-(a+d)+1-p= 2 - 2 a - 2 d - 2 j - 2 p$
    \label{gamma}
    , which depends on $p$.  We now rescale the spectral variable,
$x\to i\left(\frac{x}{2}+a\right)$, so that the factor $a+ix$ appearing in
\eqref{eq:3t-contHahn} becomes $-x/2$. Multiplying the recurrence by $-2$, the para-Krawtchouk polynomials \cite{parakrawtchouk} are defined by
\begin{equation}
    xp_n(x)=\frac{(2j+p-n)(2j+2p-2-2n+\gamma)}
{2(2n-2j-2p+1)}(p_{n+1}(x)-p_n(x))+\frac{n(2j+2-2n-\gamma)}
{2(2n-2j-1)} (p_{n-1}(x)-p_n(x))\,,
\end{equation}
with $p_{-1}=0, p_0=1.$ The monic para-Krawtchouk polynomials satisfy
\begin{equation}
xP_n(x)=P_{n+1}(x)+B_nP_n(x)+U_nP_{n-1}(x),
\qquad n=0,1,\ldots,N,
\label{monicPKrecurrence}
\end{equation}
with $P_{-1}(x)=0$, $P_0(x)=1$ and
\begin{align}
B_n
&=-\frac{(2j+p-n)(2j+2p-2-2n+\gamma)}{2(2n-2j-2p+1)}-\frac{n(2j+2-2n-\gamma)}{2(2n-2j-1)}\,,\label{recurrenceB}\\
U_n
&=\frac{n\left(2 j +p -n +1\right) \left(2 j +2 p -2 n +\gamma \right)  \left(2 j +2-2 n -\gamma \right)}{4 \left(2 n -1-2 j -2 p \right) \left(2 n -2 j -1\right)}\,.\label{recurrenceU}
\end{align}
The polynomial $P_{N+1}(x)$ is defined by \eqref{monicPKrecurrence} with $n=N$. See \eqref{PK1}-\eqref{PK3} for explicit expressions of $P_n(x)$.  The recurrence coefficients are mirror symmetric,
\begin{equation}
B_n=B_{N-n},\qquad U_n=U_{N+1-n},
\label{persymmetry}
\end{equation}
for both $p=0$ and $p=1$ and the associated Jacobi matrix is persymmetric.

\paragraph{Discrete orthogonality.}
The para-Krawtchouk polynomials satisfy the discrete orthogonality relation
\begin{equation}
\sum_{s=0}^{N}
P_n(x_s)\,P_m(x_s)\,w_s
=
h_n\,\delta_{nm},
\qquad
h_n = U_1 U_2 \cdots U_n,
\end{equation}
where the $x_s$ are the zeros of $P_{N+1}$ and the $w_s$ the corresponding
weights.  The zeros are simple and form the finite bi-lattice
\begin{align}
&x_{2s}=2s,\qquad
s=0,\ldots,j,
\label{evenspec}\\
&x_{2s+1}=2s+\gamma,
\qquad
s=0,\ldots,j-1+p.
\label{oddspec}
\end{align}
For the two sub-lattices to interlace without colliding, $\gamma$ must lie in the range 
\begin{equation}
\gamma \in (0,2).
\label{gammarange}
\end{equation} 
No other range is admissible, as to ensure the positivity of the recurrence coefficient $U_n$ for $n=j$, $N$ even. This restriction will be used again in Subsection \ref{irred}. Up to a factor fixed by the normalization condition $\sum_{s=0}^N w_s =1$,
the weights are
\begin{equation}
w_s=\frac{1}{\left|P_{N+1}'(x_s)\right|},
\end{equation}
and for $N=2j+1$ they read \cite{parakrawtchouk}, according to the parity of $s$,
\begin{equation}
w_{2s}
=
\frac{2^{-N}(1-\gamma/2)_j}{(1/2)_j}
\,
\frac{(-j)_s(-\gamma/2-j)_s}
{s!(1-\gamma/2)_s},
\qquad
s=0,1,\ldots,j,
\label{weight_even}
\end{equation}
\begin{equation}
w_{2s+1}
=
\frac{2^{-N}(1+\gamma/2)_j}{(1/2)_j}
\,
\frac{(-j)_s(\gamma/2-j)_s}
{s!(1+\gamma/2)_s},
\qquad
s=0,1,\ldots,j.
\label{weight_odd}
\end{equation}
The weights for the even case can be derived through the Christoffel transform, explained below.
\paragraph{Spectral transformation.}
Fix $j$ and $\gamma$. Here $P_n(x;p,j,\gamma)$ denotes the para-Krawtchouk polynomial with $N=2j+p$
for the given
value of $p,j$ and $\gamma$. The family with $p=0$ is obtained from the
family with $p=1$ by the Christoffel transform that removes the last spectral point $x_N=2j+\gamma$. Notice that every quantity attached to the para-Krawtchouk polynomials depends on $a,d$ only through $\gamma$. The transform therefore
requires only that the two families share the same $j$ and $\gamma$.
\begin{equation}
P_n(x;0,j,\gamma)
=
\frac{P_{n+1}(x;1,j,\gamma)-D_nP_n(x;1,j,\gamma)}{x-x_N},
\qquad
D_n=\frac{P_{n+1}(x_N;1,j,\gamma)}{P_n(x_N;1,j,\gamma)},
\end{equation}
The weights transform as $w_{s;0,j,\gamma}=\frac{\left(x_N-x_s\right)}{j+\gamma/2}w_{s;1,j,\gamma}$, the subscripts corresponding to the parameters in $P_n(x;p,j,\gamma)$.
\paragraph{Difference equation.}
The limit \eqref{parametrization} applied to the difference equation \eqref{diffc_hahn}, after
rescaling of the spectral variable  $x\to i(\frac{x}{2}+a)$ , gives
\begin{equation}
\lambda_nP_n(x)=E(x)P_n(x+2)-
\bigl(E(x)+F(x)\bigr)P_n(x)
+
F(x)P_n(x-2),
\label{diffPK}
\end{equation}
with
\begin{equation}
E(x)
=\frac12
(x-2j)(x-2j-2p+2-\gamma),
\qquad
F(x)
=\frac12x(x-\gamma),\qquad \lambda_n=2n(n-N)\,.
\label{EF}
\end{equation}
The coefficient $E(x)$ vanishes at the last two spectral points and $F(x)$ at
the first two.  The eigenvalues of the difference operator are moreover paired:
\begin{equation}
\lambda_n=\lambda_{N-n}.
\label{pairing}
\end{equation}
Equation \eqref{diffPK} is of fourth order in $x$, where the coefficients associated to the nearest shift $x+1$ and $x-1$ correspond to $0$. With respect to the bi-lattice \eqref{evenspec}--\eqref{oddspec}, the shift $x\to x+2$ maps each spectral point to the next point on the same sublattice, and likewise 
$ x\to x-2$ maps it to the preceding one, thus making the equation of fourth order \cite{parakrawtchouk}.
\paragraph{Hypergeometric representation.}
The explicit expression follows from the  limit specified in \eqref{parametrization} taken in
\eqref{hyperc_hahn}, in the variable $-x/2$. The resulting expression splits according to the
degree.  For $n\leq j$,
\begin{equation}
P_n(x)=\kappa_n\,
{}_3F_2\!\left(
\begin{matrix}
-n,\; n-2j-p,\; -x/2 \\[0.2cm]
-j,\; -j+1-p-\gamma/2
\end{matrix}
\,;1
\right),
\label{PK1}
\end{equation}
whereas for $n>j$, it involves two hypergeometric functions, one of degree at most $j$ and one of degree greater than $j$:
\begin{equation}
\begin{aligned}
P_n(x)
={}&
\kappa_n\,
{}_3F_2\!\left(
\begin{matrix}
-n,\; n-2j-p,\; -x/2 \\[0.2cm]
-j,\; -j+1-p-\gamma/2
\end{matrix}
\,;1
\right)
\\[0.5em]
&+
\kappa_n \frac{
2(-n)_{j+1}(n-2j-p)_{2j+p-n}(n-j-p)!
\left(-\dfrac{x}{2}\right)_{j+1}
}{
(j+1)!(-j)_j
\left(-j+1-p-\gamma/2\right)_{j+1}
}
\\
&\times
{}_3F_2\!\left(
\begin{matrix}
-n+j+1,\; n-j+1-p,\; -x/2+j+1\\
j+2,\; 2-p-\gamma/2
\end{matrix}
\,;1
\right).
\end{aligned}
\label{PK2}
\end{equation}
The normalization factor is
\begin{equation}
\kappa_n=
\begin{cases}
\dfrac{2^n(-j)_n\left(-j+1-p-\gamma/2\right)_n}
{(n-2j-p)_n},
& n\le j,\\[2ex]
\dfrac{(-1)^{\,n-j-p}2^{\,n-1}
\left(-j+1-p-\gamma/2\right)_n\,j!}
{(n-j)_{\,n-j-p}(2j+p-n)!},
& n>j.
\end{cases}
\label{PK3}
\end{equation}

\section{The Hahn algebra}
\label{Hahn algebra}

The connection between the para-Krawtchouk polynomials and the Hahn algebra can be seen directly at the level of their operators. The multiplication operator by the variable and the difference operator associated with the para-Krawtchouk polynomials satisfy closed commutation relations involving quadratic combinations of the two operators. These relations are those of the Hahn algebra, with the corresponding structure constants taking different values according to the parity of $N$. We give these commutation relations explicitly, together with the corresponding Casimir operator. Let $X$ be the operator of multiplication by the variable,
\begin{equation}
Xf(x)=xf(x),
\label{X}
\end{equation}
and $Y$ the difference operator in \eqref{diffPK},
\begin{equation}
Yf(x)
=
E(x)f(x+2)
+
F(x)f(x-2)
-
\bigl(E(x)+F(x)\bigr)f(x),
\label{Y}
\end{equation}
with $E$ and $F$ given by \eqref{EF}. Defining the commutator $[X,Y]=XY-YX$ and the anti-commutator $\{X,Y\}=XY+YX$, these operators satisfy
the relations of the Hahn algebra \cite{mutual}:
\begin{equation}
\begin{aligned}
[Y,[X,Y]] &= -4\{X,Y\}+4(1-N^2)X+4(N+\gamma-1)Y+M_1,\\[1ex]
[[X,Y],X] &= -4X^2-4Y+4(N+\gamma-1)X+M_2,
\end{aligned}
\label{hahnalgebra}
\end{equation}
where
\begin{equation}
\begin{aligned}
M_1 &=2(N+\gamma-2+p)(N-p)(N+1),
\qquad
M_2 =2(N+\gamma-2+p)(p-N)
\end{aligned}
\label{csts}
\end{equation}
Notice that $M_1=-(N+1)M_2$. The Casimir element is given by
\begin{equation}
Q
=
[X,Y]^2
-4\{X^2,Y\}
+(20-4N^2)X^2
-4Y^2
+4(N+\gamma-1)\{X,Y\}
+(2M_1-16(N+\gamma-1))X
+2M_2Y.
\end{equation}
In the realization \eqref{X}--\eqref{Y}, it
reduces to
\begin{equation}
q=2j(2j+\gamma-2+2p)\big((2j+p)(2j+p+\gamma-2)+(4-p)\gamma-8+p \big).
\end{equation}

\section{The parity decomposition}
\label{reducibility}

We now turn to the decomposition of the module carried by the para-Krawtchouk
polynomials, carried out first in the spectral representation, where the
block structure is made plain, and then in the recurrence representation,
where the two submodules are identified with Hahn modules and the pairing
\eqref{pairing} of the eigenvalues is accounted for.  Subsection \ref{irred}
settles the irreducibility of the submodules. The decomposition depends on the parity of $N$. In particular, the dimensions of the two submodules differ between the odd and even cases. We therefore consider the two cases separately.

\subsection{$N$ odd}

\subsubsection{Spectral representation}
\label{decomp: Spectral representation}
We first consider the spectral representation, where the decomposition into invariant subspaces is particularly transparent. Let $\mathcal{P}_N$ denote the space of polynomials of degree at most $N$, on
which $X$ acts by multiplication by the variable and $Y$ acts as the difference operator
\eqref{Y}.  A matrix realization is obtained through the weighted evaluation
map
\begin{equation}
\mathcal{U}:\mathcal{P}_N\longrightarrow\mathbb{C}^{N+1},
\qquad
\mathcal{U}f
=
\sum_{s=0}^{N}\sqrt{w_s}\,f(x_s)\ket{e_s},
\end{equation}
where $\{\ket{e_s}\}_{s=0}^{N}$ is the standard basis of $\mathbb{C}^{N+1}$. The polynomial $P_n$ is represented by the vector
\begin{equation}
\ket{n}
\coloneqq 
\mathcal{U}P_n
=
\sum_{s=0}^{N}
\sqrt{w_s}\,P_n(x_s)\ket{e_s}.
\end{equation}
With $W=\operatorname{diag}(w_0,w_1,\ldots,w_N)$, the operators $X$ and $Y$ are
represented by
\begin{equation}
X=
\begin{pmatrix}
x_0 \\
& x_1 \\
&& \ddots \\
&&& x_N
\end{pmatrix},
\qquad
Y=W^{1/2}
\begin{pmatrix}
A_0 & 0 & E_0 \\
0 & A_1 & 0 & E_1 \\
F_2 & 0 & A_2 & 0 & E_2 \\
& F_3 & 0 & A_3 & 0 & \ddots \\
&& F_4 & 0 & \ddots & \ddots \\
&&& \ddots & \ddots & 0 & E_{N-2} \\
&&&& F_N & 0 & A_N
\end{pmatrix}W^{-1/2},
\end{equation}
where $E_n=E(x_n)$, $F_n=F(x_n)$ and $A_n=-(E_n+F_n)$.  One has
\begin{align}
X\ket{n}
&=
\ket{n+1}
+
B_n\ket{n}
+
U_n\ket{n-1},
\\
Y\ket{n}
&=
\lambda_n\ket{n}.
\end{align}

The matrix of $Y$ connects only basis vectors of the same parity. 
Reordering the basis so that the even sites come first,
\begin{equation}
\ket{\nu}
=
\bigl(
\sqrt{w_0}P_n(x_0),\;
\sqrt{w_2}P_n(x_2),\;
\ldots,\;
\sqrt{w_{N-1}}P_n(x_{N-1}),\;
\sqrt{w_1}P_n(x_1),\;
\ldots,\;
\sqrt{w_N}P_n(x_N)
\bigr)^T,
\end{equation}
is equivalent to conjugating by the permutation matrix 
\begin{equation}
\Phi=
\sum_{r=1}^{j+1}
E_{r,\,2r-1}
+
\sum_{r=1}^{j+1}
E_{j+1+r,\,2r},
\label{iso}
\end{equation}
where $E_{r,s}$ is an elementary matrix of size
$(N+1)\times(N+1)$, with a single $1$ in row $r$,
column $s$, and zeros everywhere else:
\[
(E_{r,s})_{kl}=\delta_{kr}\delta_{ls}.
\]
$\Phi$ sends $(x_0,x_1,\ldots,{x_{N}})^T$
to $(x_0,x_2,\ldots,x_{N-1},
x_1,x_3,\ldots,x_{N})^T$.  Writing
$\widetilde X=\Phi X\Phi^{-1}$ and $\widetilde Y=\Phi Y\Phi^{-1}$,
 the relations
\begin{equation}
\widetilde{X}\ket{\nu}
=
\ket{\nu+1}
+
B_n\ket{\nu}
+
U_n\ket{\nu-1},
\qquad
\widetilde{Y}\ket{\nu}
=
\lambda_n\ket{\nu}\,,
\end{equation}
are unchanged.
Since $E_n$ vanishes for $n=N-1,N$ and $F_n$ for $n=0,1$, the reordering
produces zeros in the off-diagonal entries of $\widetilde{Y}$ at the junction
of its two halves, so that $\widetilde Y$ is block diagonal with two
tridiagonal blocks.  As $\widetilde X$ is diagonal it preserves this
decomposition.  One gets $\widetilde{W}=\operatorname{diag}(w_0,w_2,\ldots,w_{N-1},w_1,\ldots,w_N)$ and
\begingroup
\renewcommand{\arraystretch}{1.05}
\setlength{\arraycolsep}{3pt}
\begin{equation}
\widetilde X=
\begin{pmatrix}
\blk{xa}{x_0} \\
& x_2 \\
&& \ddots \\
&&& \blk{xb}{x_{N-1}}\hspace{2pt} \\[4pt]
&&&& \hspace{2pt}\blk{xc}{x_1} \\
&&&&& x_3 \\
&&&&&& \ddots \\
&&&&&&& \blk{xd}{x_N}
\end{pmatrix},
\label{tildeX}
\end{equation}
\begin{tikzpicture}[remember picture, overlay]
\coordinate (xmid) at ($(xb)!0.5!(xc)$);
\draw[dashed, thick] (xa.west -| xmid) -- (xd.east -| xmid);
\draw[dashed, thick] (xmid -| xa.north) -- (xmid -| xd.south);
\coordinate (xm1) at ($(xa)!0.5!(xb)$);
\coordinate (xm2) at ($(xc)!0.5!(xd)$);
\node at (xm1 -| xm2) {\Large $0$};
\node at (xm2 -| xm1) {\Large $0$};
\end{tikzpicture}
\endgroup

\renewcommand{\arraystretch}{1.35}
\setlength{\arraycolsep}{7pt}

\begin{equation}
\widetilde Y=\widetilde{W}^{1/2}
\begin{pmatrix}
\blk{ya}{A_0} & E_0 \\
F_2 & A_2 & E_2 \\
& \ddots & \ddots & \ddots \\
&& F_{N-1} & \blk{yb}{A_{N-1}} \\
&&&& \blk{yc}{A_1} & E_1 \\
&&&& F_3 & A_3 & E_3 \\
&&&&& \ddots & \ddots & \ddots \\
&&&&&& F_N & \blk{yd}{A_N}
\end{pmatrix}
\widetilde{W}^{-1/2}
\tikz[remember picture,overlay]{
    \coordinate (ymid) at ($(yb)!0.5!(yc)$);

    \draw[dashed,thick]
        (ya.west -| ymid)
        --
        (yd.east -| ymid);

    \draw[dashed,thick]
        (ymid -| ya.north)
        --
        (ymid -| yd.south);

    \coordinate (ym1) at ($(ya)!0.5!(yb)$);
    \coordinate (ym2) at ($(yc)!0.5!(yd)$);

    \node at (ym1 -| ym2) {\Large $0$};
    \node at (ym2 -| ym1) {\Large $0$};
}
,
\end{equation}
The even and odd subspaces are therefore invariant under both
$\widetilde{X}$ and $\widetilde{Y}$ and the module carried by the
para-Krawtchouk polynomials is reducible. To identify the two blocks, set
\begin{subequations}
\begin{align}
P_n^e(s)
&=
P_n(x_{2s}),
\qquad
P_n^o(s)
=
P_n(x_{2s+1}),
\label{eq:transformed-polynomials-a}
\\[4pt]
E^e(s)
&=
\frac{E(x_{2s})}{2}
=
(s-j)\left(s-j-\frac{\gamma}{2}\right),
\qquad
F^e(s)
=
\frac{F(x_{2s})}{2}
=
s\left(s-\frac{\gamma}{2}\right),
\label{eq:transformed-polynomials-c}
\\[4pt]
E^o(s)
&=
\frac{E(x_{2s+\gamma})}{2}
=
(s-j)\left(s-j+\frac{\gamma}{2}\right),
\qquad
F^o(s)
=
\frac{F(x_{2s+\gamma})}{2}
=
s\left(s+\frac{\gamma}{2}\right).
\label{eq:transformed-polynomials-d}
\end{align}
\end{subequations}
The difference equation \eqref{diffPK}, restricted to either sub-lattice and divided by
two, becomes
\begin{equation}
n(n-N)P_n^e(s)
=
E^e(s)\,P_n^e(s+1)
+
F^e(s)\,P_n^e(s-1)
-
\bigl(E^e(s)+F^e(s)\bigr)P_n^e(s),\label{eq:transformed-diff-pkrawe}
\end{equation}
\begin{equation}
n(n-N)P_n^o(s)
=
E^o(s)\,P_n^o(s+1)
+
F^o(s)\,P_n^o(s-1)
-
\bigl(E^o(s)+F^o(s)\bigr)P_n^o(s)\label{eq:transformed-diff-pkrawo},
\end{equation}
both of which are the difference equation of the Hahn polynomials \eqref{hahndiffeq}. Comparison of
\eqref{eq:transformed-polynomials-c} and
\eqref{eq:transformed-polynomials-d} with \eqref{hahndiffeq} gives $K=j$ and
\begin{equation}
\alpha_e=-(j+1)-\frac{\gamma}{2},
\qquad
\beta_e=-(j+1)+\frac{\gamma}{2},
\label{oparame}
\end{equation}
\begin{equation}
\alpha_o=-(j+1)+\frac{\gamma}{2},
\qquad
\beta_o=-(j+1)-\frac{\gamma}{2}.
\label{oparamo}
\end{equation}
In both cases $\alpha+\beta+1=-N$, so that the Hahn eigenvalue 
$n(n+\alpha+\beta+1)$ of the difference equation \eqref{hahndiffeq} is $n(n-N)$ matching \eqref{eq:transformed-diff-pkrawe}-- \eqref{eq:transformed-diff-pkrawo}.

\subsubsection{Recurrence representation}
\label{Recrep}

While the spectral picture reveals the block structure of the system, the
recurrence picture provides a more explicit description of the direct sum
decomposition and, in particular, explains the pairing \eqref{pairing} between
the eigenvalues. To make this connection precise, let
\begin{equation}
   Q^e_n(s)=Q_n(s;\alpha_e,\beta_e,j),
   \qquad
Q^o_n(s)=Q_n(s;\alpha_o,\beta_o,j), \qquad s=0,1,\ldots,j
\end{equation}
where $Q_n$ is the standard Hahn polynomial \eqref{eq:Hahn-pol-dfn}, and the parameters $\alpha_e, \beta_e, \alpha_o,\beta_o$ are defined in \eqref{oparame} and \eqref{oparamo}. We write 
$p^i_n(s)=\kappa_n^{-1}P^i_n(s)$ for the normalized para-Krawtchouk polynomials \eqref{eq:transformed-polynomials-a}. The Sears transformation formula (see relation (III.15) in \cite{GasperRahman2004}) takes the following form for a terminating ${}_3F_2$:
\begin{equation}
{}_3F_2\left(
\begin{matrix}
-m,\; a,\; b\\
c,\; d
\end{matrix}
;1
\right)
=
\frac{(c-a)_m}{(c)_m}
{}_3F_2\left(
\begin{matrix}
-m,\; a,\; d-b\\
a-c-m+1,\; d
\end{matrix}
\,\,;1
\right).
\label{eq:Sears3F2}
\end{equation}
\begin{proposition}
\label{eq:sim-antisim-parakraw}
The following symmetric and antisymmetric combinations hold:
\begin{align}
    Q^e_n(s)&=\tfrac{1}{2}\bigl(p^e_n(s)+p^e_{N-n}(s)\bigr),
\label{evenhahn}\\
    Q^o_n(s)&=\frac{(-j-\gamma/2)_n}{2(-j+\gamma/2)_n}
    \bigl(p^o_n(s)-p^o_{N-n}(s)\bigr).
    \label{oddhahn}
\end{align}    
\end{proposition}
\begin{proof}
The pairing of eigenvalues is the key mechanism underlying these combinations: $p_n$ and $p_{N-n}$ correspond to the same eigenvalue of $Y$, and can therefore be combined within the corresponding eigenspace. For \eqref{evenhahn}, inputting the even spectral points directly cancels the part of the hypergeometric expression corresponding to degrees higher than $j$, leaving two identical hypergeometric expressions whose sum gives twice the even Hahn polynomial. For \eqref{oddhahn}, the subtraction removes the part of the hypergeometric expression corresponding to degrees lower than $j$, leaving $-1$ times the second half of \eqref{PK2}. Working backwards from the associated odd Hahn polynomial, four successive applications of the Sears transformation \eqref{eq:Sears3F2} yield the desired result, with
\begin{align}
    &\text{Iteration 1:}\quad  m=n,\quad a=n-N,\quad b=-s,\quad c=-j+\frac{\gamma}{2},\quad d=-j,\\
    &\text{Iteration 2:}\quad   m=N-n,\quad a=s-j,\quad b=-n,\quad c=-j-\frac{\gamma}{2},\quad d=-j,\\
    &\text{Iteration 3:} \quad  m=j-s,\quad a=n-j,\quad b=n-N,\quad c=s-2j+n+\frac{\gamma}{2},\quad d=-j,\\
    &\text{Iteration 4:}\quad m=j-n,\quad a=j+1-n,\quad b=s-j,\quad c=-j,\quad d=1-\frac{\gamma}{2}.
\end{align}
Adjusting the coefficients then gives the remaining half of the expression for the para-Krawtchouk polynomials with $n>j$.
\end{proof}
The
combinations \eqref{evenhahn}--\eqref{oddhahn} separate the sub-lattices in a strong sense, each vanishing
identically on the sub-lattice the other lives on:
\begin{equation}
p^o_n(s)+p^o_{N-n}(s)=0,
\qquad
p^e_n(s)-p^e_{N-n}(s)=0,
\label{parityvanish}
\end{equation}
for all $n$ and $s$.  This is what makes \eqref{evenhahn} and
\eqref{oddhahn} possible on one sub-lattice at a time.
Next we write the combinations \eqref{evenhahn} and \eqref{oddhahn} in matrix form. 
Introduce the transition matrices
\begin{align}
    &T^{PK}=
\begin{pmatrix}
p^e_0(0) & \cdots & p^e_0(j) &
p^o_0(0) & \cdots & p^o_0(j) \\
\vdots & \ddots & \vdots &
\vdots & \ddots & \vdots \\
p^e_N(0) & \cdots & p^e_N(j) &
p^o_N(0) & \cdots & p^o_N(j)
\end{pmatrix}\,,\quad T^{\mathcal{Q}^{i}}
=
\begin{pmatrix}
Q^{i}_0(0) & \cdots & Q^{i}_0(j) \\
\vdots & \ddots & \vdots \\
Q^{i}_j(0) & \cdots & Q^{i}_j(j)
\end{pmatrix}\,.
\end{align}
The first $j+1$ columns of $T^{PK}$ correspond to the even sub-lattice, while the
last $j+1$ columns correspond to the odd sub-lattice. The index $i\in\{e,o\}$ in $T^{\mathcal{Q}^{i}}$ labels the
even and odd cases, respectively. To implement the relations \eqref{evenhahn} and \eqref{oddhahn}, we introduce the matrix
\begin{equation}
\mathcal{C}
=\frac12
\begin{pmatrix}
I_{j+1} & R \\
DR & -D
\end{pmatrix},
\end{equation}
where  $I_{j+1}$ denotes the identity matrix of dimension $j+1$,  $D$ is a diagonal matrix, and $R$ is a reflection matrix, defined as
follows:
\begin{equation}
D=
\operatorname{diag}\left(
\frac{\left(-j-\frac{\gamma}{2}\right)_j}{\left(-j+\frac{\gamma}{2}\right)_j},
\ldots,
\frac{\left(-j-\frac{\gamma}{2}\right)_0}{\left(-j+\frac{\gamma}{2}\right)_0}
\right),
\qquad
R=
\begin{pmatrix}
 & & 1\\
& \iddots & \\
1 & &
\end{pmatrix}.
\label{reflexmatrix}
\end{equation}
The symmetric and antisymmetric combinations \eqref{evenhahn}--\eqref{oddhahn} can be written simultaneously, using \eqref{parityvanish}, as
\begin{equation}
    \mathcal{C}T^{PK}=
\begin{pmatrix}
T^{\mathcal{Q}^{e}}&0\\0 &
RT^{\mathcal{Q}^{o}}
\end{pmatrix}:=T^\mathcal{Q}.
\label{transitionhahn}
\end{equation}
\paragraph{Block decomposition of the operators.}
Let $\widehat{J}^{PK}$ denote the matrix representing multiplication by $x$
in the basis $\{p_n\}$, so that 
$\widehat{J}^{PK}T^{PK}=T^{PK}\widetilde{X}$
where $\widetilde{X}$ is the diagonal matrix given by \eqref{tildeX}.
It can be written as 
\begin{equation}
\widetilde{X}
=\begin{pmatrix}
2\Lambda & \\
& 2\Lambda+\gamma I_{j+1}
\end{pmatrix},
\qquad
\Lambda=\operatorname{diag}(0,1,\ldots,j),
\label{Lambdatilde}
\end{equation}
Here, the diagonal entries of $\Lambda$ are precisely the points of the Hahn grid $s=0,\ldots,j$.  Hence, the two sub-lattices are affine images of
the Hahn grid: the even sub-lattice is obtained by a dilation by $2$, while
the odd sub-lattice is obtained by the same dilation followed by a shift by
$\gamma$.  Writing
$\widehat{J}^{Q^e}$ and $\widehat{J}^{Q^o}$ for the matrices of the
multiplication by $s$ in the bases $\{Q^e_n\}$ and $\{Q^o_n\}$, so that
$\widehat{J}^{Q^i}T^{\mathcal{Q}^{i}}=T^{\mathcal{Q}^{i}}\Lambda$, we get from
\eqref{transitionhahn}
\begin{equation}
\mathcal{C}\,\widehat{J}^{PK}\mathcal{C}^{-1}=\mathcal{C}T^{PK}\widetilde{X}(T^{PK})^{-1}\mathcal{C}^{-1}
=T^{\mathcal{Q}}\widetilde{X}\left(T^{\mathcal{Q}}\right)^{-1}
=\begin{pmatrix}
2\widehat{J}^{Q^e} & 0\\
0& 2R\widehat{J}^{Q^o}R+\gamma I
\end{pmatrix}.
\label{Jconj}
\end{equation}
Equation \eqref{Jconj} is the recurrence-representation counterpart of the block decomposition found for $\widetilde{Y}$ in Subsection \ref{decomp: Spectral representation}. 
The difference operator $Y$ is represented by $
Y=\operatorname{diag}(\lambda_0,,\ldots,\lambda_j,\lambda_{j+1},\ldots,\lambda_N).
$ By the eigenvalue pairing $\lambda_n=\lambda_{N-n}$, $Y=\operatorname{diag}(\lambda_0,,\ldots,\lambda_j,\lambda_{j},\ldots,\lambda_0) $.
Conjugation by $\mathcal{C}$ leaves this matrix
unchanged $
\mathcal{C}\,Y\,\mathcal{C}^{-1}=Y.$

Inverting \eqref{evenhahn} and \eqref{oddhahn} expresses the para-Krawtchouk
polynomials through the polynomials of the two submodules:
\begin{align}
p^e_n(s)=
\begin{cases}
Q_n^e(s), & n\le j,\\[0.8em]
Q_{N-n}^e(s), & n>j,
\end{cases}\qquad
p^o_n(s)=
\begin{cases}
\dfrac{\left(-j+\gamma/2\right)_n}
{\left(-j-\gamma/2\right)_n}
Q_n^o(s), & n\le j,\\[1.8em]
-\dfrac{\left(-j+\gamma/2\right)_{N-n}}
{\left(-j-\gamma/2\right)_{N-n}}
Q_{N-n}^o(s), & n>j.
\end{cases}
\end{align}

\subsection{$N$ even}
For $N=2j$ the argument runs as before.  The coefficients $E(x)$ and $F(x)$ vanish at the last two and the first two grid points respectively. Moreover,  reordering the basis by parity makes the difference operator block
diagonal, while the diagonal operator $X$ preserves the two blocks. In this case, however, the two submodules no longer
have the same dimension: the even one is of dimension $j+1$ and the odd one of
dimension $j$.  The corresponding parameters are
\begin{align}
&\widetilde{\alpha}_e=-j-\gamma/2,\qquad
\widetilde{\beta}_e=-(j+1)+\gamma/2,
\qquad K=j,
\label{eparame}\\
&\widetilde{\alpha}_o=-(j+1)+\gamma/2,\qquad
\widetilde{\beta}_o=-j-\gamma/2,
\qquad K=j-1,
\label{eparamo}
\end{align}
and the corresponding analogues of \eqref{evenhahn} and \eqref{oddhahn} are
\begin{align}
    &\widetilde{Q}^e_n(s)=\tfrac{1}{2}
    \bigl(\widetilde{p}^e_n(s)+\widetilde{p}^e_{N-n}(s)\bigr),
\label{evenhahneven}\\
   & \widetilde{Q}^o_n(s)=
    \frac{j\left(-j+1-\gamma/2\right)_n}
    {2(j-n)\left(-j+\gamma/2\right)_n}
    \bigl(\widetilde{p}^o_n(s)-\widetilde{p}^o_{N-n}(s)\bigr),
    \label{oddhahneven}
\end{align}
where
$\widetilde{Q}^e_n=Q_n(x;\widetilde{\alpha}_e,\widetilde{\beta}_e,j)$,
$\widetilde{Q}^o_n=Q_n(x;\widetilde{\alpha}_o,\widetilde{\beta}_o,j-1)$
and $\widetilde{p}_n(x)=\kappa_n^{-1}\widetilde{P}_n(x)$, the tilde used to distinguish even from odd.  Inverting these relations expresses the para-Krawtchouk polynomials in terms of the polynomials of the two submodules:
\begin{align}
\widetilde{p}^e_n(s)=
\begin{cases}
\widetilde{Q}_n^e(s), & n\le j,\\[0.8em]
\widetilde{Q}_{N-n}^e(s), & n>j,
\end{cases}\qquad
\widetilde{p}^o_n(s)=
\begin{cases}
\dfrac{(j-n)\left(-j+\gamma/2\right)_n}
{j\left(-j+1-\gamma/2\right)_n}
\widetilde{Q}_n^o(s), & n\le j,\\[1.8em]
-\dfrac{(n-j)\left(-j+\gamma/2\right)_{N-n}}
{j\left(-j+1-\gamma/2\right)_{N-n}}
\widetilde{Q}_{N-n}^o(s), & n>j.
\end{cases}\label{eq:p_n-x2sp1}
\end{align}
\begin{remark}
    For even $N$, the zeros of the polynomial $p_j$ form one half of the grid. Indeed, \eqref{eq:p_n-x2sp1} shows that, when $n=j$, the zeros of $p_j$ coincide with the odd grid points $x_{2s+1}$. This phenomenon does not occur when $N$ is odd. It can also be verified directly from \eqref{PK1}: the corresponding ${}_3F_2$ reduces to a ${}_2F_1$
\begin{equation}
P_j\!\left(x_{2s+1}\right)
=
\kappa_j\,
{}_2F_1\!\left(
\begin{matrix}
-j,\; -x_{2s+1}/2\\
-j+1-\gamma/2
\end{matrix}\,\,
;1
\right).
\end{equation}
Applying the Chu-Vandermonde transformation
\begin{equation}
{}_2F_1\!\left(
\begin{matrix}
-j,\; b\\
c
\end{matrix}
;1
\right)
=
\frac{(c-b)_j}{(c)_j},
\end{equation}
we obtain
\begin{equation}
P_j(x_{2s+1})
=
\kappa_j\,
\frac{(s-j+1)_j}
{\left(-j+1-\gamma/2\right)_j}.
\end{equation}
With $s< j$, the expression becomes 0. 
\end{remark}

\subsection{Irreducibility of the submodules}
\label{irred}

It remains to check that the two summands cannot be reduced further.  For the
Hahn algebra this is settled by a direct criterion.

\begin{proposition}
\label{prop:hahnirred}
Let the Hahn algebra act on the functions on the grid $y_s=s$,
$s=0,1,\ldots,K$, with $X$ diagonal and $Y$ the difference operator
\eqref{hahndiffeq} with coefficients \eqref{hahndiffcoeff}.  The module so obtained is
irreducible if and only if
\begin{equation}
\alpha\notin\{-1,-2,\ldots,-K\}
\qquad\text{and}\qquad
\beta\notin\{-1,-2,\ldots,-K\}.
\label{irredcond}
\end{equation}
\end{proposition}

\begin{proof}
The operator $X$ is diagonal with the $K+1$ distinct eigenvalues $y_s=s$, so
any invariant subspace is spanned by a subset $S$ of the basis vectors
$\ket{e_s}$.  In that basis $Y$ is tridiagonal, with entries $B(y_s)$ above
the diagonal and $D(y_s)$ below it.  The subspace spanned by $S$ is invariant
under $Y$ if and only if $s\in S$ and $B(y_s)\neq0$ imply $s+1\in S$, and
$s\in S$ and $D(y_s)\neq0$ imply $s-1\in S$.  A proper invariant subspace
therefore exists if and only if one of these coefficients vanishes strictly
inside the grid.  Now $B(y_s)=(s+\alpha+1)(s-K)$ vanishes for some
$s\in\{0,\ldots,K-1\}$ exactly when $\alpha\in\{-1,\ldots,-K\}$, and
$D(y_s)=s(s-\beta-K-1)$ vanishes for some $s\in\{1,\ldots,K\}$ exactly when
$\beta\in\{-K,\ldots,-1\}$.  The two boundary zeros $B(y_K)=0$ and $D(y_0)=0$
are those that make the grid finite and do not split the module.
\end{proof}

Each of the four parameters \eqref{oparame}, \eqref{oparamo}, \eqref{eparame} and \eqref{eparamo} governing the submodules is of the form $m\pm \frac{\gamma}{2}$ for some integer $m$. The restriction $\gamma \in (0,2)$, imposed in \eqref{gammarange} so that the two sub-lattice interlace without colliding, gives $\frac{\gamma}{2}\in(0,1)$ so that none of these parameters is ever an integer: condition \eqref{irredcond} is met for every admissible $\gamma$, and Proposition \ref{prop:hahnirred} shows that the two summands are irreducible.  Collecting the results of this section:

\begin{theorem}
\label{thm:PK}
Let $N = 2j+p$ with $p=0,1$ according to the parity of $N$, and let
$\gamma \in (0,2)$.  The $(N+1)$-dimensional module of the Hahn algebra
carried by the para-Krawtchouk polynomials $V_{N+1}^{\mathrm{PK}}$, is the direct sum of two
irreducible Hahn modules,
\begin{equation}
  V_{N+1}^{\mathrm{PK}} \cong Q(\alpha_e, \beta_e, j) \oplus
  Q(\alpha_o, \beta_o, j),   \qquad p=1,
\end{equation}
\begin{equation}
  V_{N+1}^{\mathrm{PK}} \cong Q(\widetilde{\alpha}_e, \widetilde{\beta}_e, j)
  \oplus Q(\widetilde{\alpha}_o, \widetilde{\beta}_o, j-1),   \qquad p=0,
\end{equation}
where $Q(\alpha, \beta, K)$ denotes the $(K+1)$-dimensional Hahn module with
parameters $\alpha$ and $\beta$, and the parameters are those of
\eqref{oparame}, \eqref{oparamo}, \eqref{eparame} and \eqref{eparamo}.
\end{theorem}

\begin{proof}
Subsections \ref{decomp: Spectral representation} and \ref{Recrep} establish
the decomposition and identify the parameters of the two summands. That each summand is irreducible then follows from Proposition \ref{prop:hahnirred}, since the parameters \eqref{oparame}--\eqref{eparamo} satisfy \eqref{irredcond} whenever $\gamma\in(0,2).$
\end{proof}

\section{The mechanism}
\label{GM}

The properties established in Section \ref{reducibility} can be extended to all polynomials of the para family. Three features of the setting are responsible for the decomposition, and identifying them is enough to carry out the reduction in general. This section focuses on this mechanism and analyzes its presence in the para-Racah, $q$-para-Racah and para-Bannai--Ito polynomials.
\subsection{Three features}
\label{mechanism}

\begin{proposition}
\label{prop:mechanism}
Let $\{p_n\}_{n=0}^{N}$ be a finite family of polynomials orthogonal on a
grid $\{x_s\}_{s=0}^{N}$. Let $X$ be the operator of multiplication by the
variable and let $Y$ be the difference or Dunkl-difference operator 
that the family diagonalizes.  Suppose that:
\begin{enumerate}
    \item[(i)] the grid is the union of two interlaced sub-lattices,
    $\{x_{2s}\}$ and $\{x_{2s+1}\}$;
    \item[(ii)] $Y$ shifts within a sub-lattice, that is, it carries a
    function supported on one sub-lattice to a function supported on the same
    sub-lattice;
    \item[(iii)] the shift coefficients of $Y$ vanish at the two ends of each
    sub-lattice.
\end{enumerate}
Then each sub-lattice spans a submodule of the module generated by $X$ and
$Y$, and that module is the direct sum of the two.
\end{proposition}

\begin{proof}
By (i) the basis of the module splits in two according to the parity of the
site.  By (ii) the matrix of $Y$ in that basis has no entry joining the two
halves, except possibly at the junction, and by (iii) those entries vanish.
After reordering the basis by parity, $Y$ is block
diagonal.  The operator $X$ is diagonal in the same basis and preserves the
blocks.  The two subspaces are therefore invariant under both generators, and
they span the module.
\end{proof}

The same argument settles the irreducibility of the two summands, once and
for all.

\begin{proposition}
\label{prop:irredgen}
Under the hypotheses of Proposition \ref{prop:mechanism}, the submodule
carried by a sub-lattice is irreducible if and only if the shift coefficients
of $Y$ vanish at no interior point of that sub-lattice.
\end{proposition}

\begin{proof}
On the sub-lattice, $X$ is diagonal with distinct eigenvalues, so any
invariant subspace is spanned by a subset of the basis vectors attached to
the spectral points.  After the reordering of Proposition
\ref{prop:mechanism}, $Y$ is tridiagonal on that basis, and such a subset is
invariant under $Y$ precisely when a shift coefficient vanishes between two
of its points.  The coefficients that vanish at the two ends are those of
condition (iii); they make the sub-lattice finite and do not split it.
\end{proof}

Proposition \ref{prop:hahnirred} above is this criterion written out for the
Hahn algebra, where it takes the form of a condition on $\alpha$ and $\beta$.
The same is done for the Racah, $q$-Racah and Bannai-Ito cases in Subsection \ref{irredPR}, \ref{irredqPR} and \ref{irredPBI}.
\begin{remark}
    Proposition \ref{prop:mechanism} gives the reduction, not the pairing of the eigenvalues.
The spectrum of $Y$ satisfies $\lambda_n=\lambda_{N-n}$ because of the truncation condition applied to the polynomials of the para family: in the para-Krawtchouk case, for instance, the truncation condition appears exactly in the expression of the continuous Hahn eigenvalue:
\begin{equation}
    \lambda_n=n(n+a+b+c+d-1),
\end{equation}
and imposing \eqref{truncation} yields the pairing \eqref{pairing}.  The same mechanism recurs for each of the other para families discussed bellow. The truncation condition of the parent infinite family always forces this same degeneracy in the difference or Dunkl-difference operator. This provides a general explanation as to why the polynomials of the para family are all reducible: the truncation always forces a degeneracy in the spectrum, allowing symmetric and antisymmetric combinations of $p_n$ and $p_{N-n}$ to be formed, and it is precisely these combinations that separate the two sub-lattices, as in \eqref{evenhahn} and \eqref{oddhahn}, thereby explaining the reduction of the recurrence picture.
\end{remark}

The remainder of this section carries the argument out for the three other para families.  Since the para-Krawtchouk polynomials are a limiting case of the para-Racah polynomials, we begin by making that limit explicit.

\subsection{The para-Racah polynomials}
\label{PRsec}
The para-Racah polynomials arise through a truncation process applied to the Wilson polynomials \cite{pararacah}. We review the procedure and collect the properties of the polynomials it produces. The Wilson polynomials, which depend on four parameters $
a,b,c,d$, are recalled in Appendix \ref{app:wilson} 
\paragraph{Truncation condition.}
The Wilson polynomials defined in $\widetilde W_n(x^2;a,b,c,d)$ \eqref{wilsondef}, have recurrence coefficients $A_n,C_n$ given by \eqref{wilsonA}--\eqref{wilsonC}. The recurrence relation \eqref{wilsonrec} terminates after a finite number of steps when 
\begin{equation}
    A_NC_{N+1}=0\,.
\end{equation}
Besides the standard truncation, which leads to the Racah polynomials, the family also truncates to a finite set of polynomials when
\begin{equation}
1-(a+b+c+d)=N\in\mathbb{Z}_{>0}.
\label{racahtrunc}
\end{equation}
It is this choice of parameters that gives rise to the para-Racah polynomials.
\paragraph{Parametrization.}
Let $N=2j+p$, where $p=0,1$, according as $N$ is even or odd, respectevely, and $j$ a positive integer. The condition \eqref{racahtrunc} is then implemented  by 
\begin{equation}
b=-a-j+e_1t,
\quad
d=(1-p)-c-j+e_2t,
\quad t\to0\,,\label{racahparam}
\end{equation}
where $e_1,e_2$ are deformation parameters. The deformation parameters $e_1,e_2$ that occur in \eqref{racahparam} are not independent and lead to a single deformation parameter $\sigma$, defined by
\begin{equation}
\sigma=\frac{e_1}{e_1+e_2},
\qquad
1-\sigma=\frac{e_2}{e_1+e_2}.
\label{sigmadef}
\end{equation}
The deformation parameter is written
$\sigma$ here and $\alpha$ in \cite{pararacah,qPR,paraBI}.
\paragraph{Recurrence relation.}
Substituting \eqref{racahparam} into the  recurrence coefficients of the Wilson polynomials \eqref{wilsonA}--\eqref{wilsonC}, and letting $t\to0$, we obtain the recurrence relation that defines the monic para-Racah polynomials $R_n(x^2)$
\begin{equation}
x^2R_n(x^2)=R_{n+1}(x^2)+b_nR_n(x^2)+u_nR_{n-1}(x^2),
\qquad n=0,1,\ldots,N,
\label{racahrecurrence}
\end{equation}
with $R_{-1}(x^2)=0$ and $R_0(x^2)=1$, with coefficients given by
\begin{align}
b_n
={}&
-\frac{p}{2}
\left[
a(a+j)+c(c+j)+n(2j+p-n)
\right]
\notag\\
&+(1-p)\Bigg[
\frac{(n-2j-p)(n+a+c)(n+a-c-j+1)}
     {2(2n+1-2j-p)}
\notag\\
&\hspace{2.5cm}
+\frac{n(n-2j-p-a-c)(n-j-1+c-a)}
     {2(2n-1-2j-p)}
-a^2
\Bigg],\label{eq:bn-pararacah}
\\
u_n
={}&
p\,
\frac{
n(2j+p+1-n)(2j+p-n+a+c)(n-1+a+c)
\left[(n-j-1)^2-(a-c)^2\right]
}{
4(2j+p-2n)(2j+p-2n+2)
}
\notag\\
&+(1-p)\,
\frac{
n(2j-n+1)(a+c+n-1)(a-c+j-n+1)
(-a+c+j-n)(a+c+2j-n)
}{
4(2j-2n+1)^2
}.\label{eq:un-pararacah}
\end{align}
For values $(p_0,n_0)$ of $(p,n)$ for which \eqref{eq:bn-pararacah} or \eqref{eq:un-pararacah} are indeterminate, one should take the limit $p \rightarrow p_0$ followed by $n \rightarrow n_0$.
The positivity of the recurrence coefficients, $u_n>0$ for $n=1,\ldots,N$,
 restricts $a$, $c$ and $\sigma$ to one of
 
\begin{equation}
\left\{
\begin{aligned}
-j-p &< a+c < -j+1,\\
a+c &\neq -j,\\
c-a &> j
\quad \text{or} \quad
c-a < 1-p-j,\\
0 &< \alpha < 1,
\end{aligned}
\right.
\qquad \text{or} \qquad
\left\{
\begin{aligned}
a+c &< -N+1
\quad \text{or} \quad
0<a+c,\\
-p &< c-a < 1,
\qquad c-a\neq 0,\\
0 &< \alpha < 1.
\end{aligned}
\right.
\label{conditions}
\end{equation}
When $\sigma=1/2$, the coefficients become mirror-symmetric, with $b_n=b_{N-n}$ and $u_n=u_{N-n+1}$ for all $n$.
\paragraph{Discrete orthogonality.} 
The para-Racah polynomials satisfy the discrete orthogonality
\begin{equation}
\sum_{s=0}^NR_n(x_s^2)R_m(x_s^2)w_s=h_n\delta_{nm}\,,\qquad h_n=u_1u_2\ldots u_n\,,
\end{equation}
where $x_s^2$ are zeros of $R_{N+1}(x^2)$, which define
the quadratic bi-lattice
 \begin{align}   
&x_{2s}^{\,2}=-(s+a)^2,\qquad
s=0,\ldots,j,\label{racahgrid-even}
\\
&x_{2s+1}^{\,2}=-(s+c)^2,
\qquad
s=0,\ldots,j-1+p.\label{racahgrid-odd}
\end{align}
To ensure that there are no collisions between the two sub-lattices, these additional conditions have to be met:
\[
a+c\notin\{-N+1,\ldots,0\} \qquad\text{and}\qquad a-c\notin\{-j,\ldots, j-1+p\}.
\]
Explicit expressions for the weights are given in \cite[eq.~(3.16)]{pararacah} for $N=2j+1$ and \cite[eq.~(4.13)]{pararacah} for $N=2j$.

\paragraph{Difference equation.} 
By inserting the parametrization \eqref{racahparam} in the diﬀerence equation of the Wilson polynomials \eqref{wilsondiff}, it is
easily seen that the limit $t\to 0 $ is trivial since there are no parameters in the denominator. Hence, the
para-Racah polynomials obey the same diﬀerence equation as the Wilson polynomials:
\begin{equation}
-n(N-n)R_n(x^2)
=
\overline{D}(x)\,R_n\!\left((x+i)^2\right)
-
\bigl(\overline{D}(x)+D(x)\bigr)R_n(x^2)
+
D(x)\,R_n\!\left((x-i)^2\right),
\label{Y_para}
\end{equation}
with
\begin{equation}
D(x)
=
\frac{
(a+ix)(-a-j+ix)(c+ix)(1-p-c-j+ix)
}{
(2ix)(2ix+1)
}
\end{equation}
and $\overline{D}(x)=D(-x)$ its complex conjugate.
As in the para-Krawtchouk case, the spectrum is doubly degenerate: $R_n$ and $R_{N-n}$ share the eigenvalue $-n(N-n)$ (cf.\ Proposition~\ref{prop:mechanism}).

\paragraph{Hypergeometric representation.}
Their explicit expression splits according to the degree.  In the case $n\leq j$, we have
\begin{equation}
R_n(x^2)
=
\eta_n\,
{}_4F_3\!\left(
\begin{matrix}
-n,\; n-2j-p,\; a-ix,\; a+ix\\
-j,\; a+c,\; a-c-j+1-p
\end{matrix}
;1
\right),
\end{equation}
whereas for $n>j$ we have
\begin{equation}
\begin{aligned}
R_n(x^2)
={}&
\eta_n\,
{}_4F_3\!\left(
\begin{matrix}
-n,\; n-2j-p,\; a-ix,\; a+ix\\
-j,\; a+c,\; a-c-j+1-p
\end{matrix}
;1
\right)
\\[4pt]
&+
\eta_n
\frac{
(-n)_{j+1}
(n-2j-p)_{2j+p-n}
(a-ix)_{j+1}
(a+ix)_{j+1}
(1)_{n-j-p}
}{
\sigma(1)_{j+1}
(-j)_j
(a+c)_{j+1}
(a-c-j+1-p)_{j+1}
}
\\[4pt]
&\times
{}_4F_3\!\left(
\begin{matrix}
-n+j+1,\; n-j+1-p,\; a+j+1-ix,\; a+j+1+ix\\
j+2,\; a+c+j+1,\; a-c+2-p
\end{matrix}
;1
\right),
\end{aligned}
\end{equation}
where $\eta_n$ is a normalization factor given in \cite{pararacah} that makes the polynomials monic.

\subsubsection{The para-Krawtchouk polynomials as a limit}
The para-Racah  polynomials that we have introduced here are orthogonal on a quadratic bi-lattice given
by \eqref{racahgrid-even}--\eqref{racahgrid-odd}. It is possible to further deform these bi-lattices into a linear
bi-lattice. Set $\sigma=1/2$, for which the recurrence coefficients in \eqref{racahrecurrence} satisfy $b_n=b_{N-n}$. We can then reparametrize
\begin{equation}
a(\theta)=\frac{\theta-\frac{\gamma}{2}}{2},
\qquad
c(\theta)=\frac{\theta+\frac{\gamma}{2}}{2},
\end{equation}
the quadratic bi-lattice \eqref{racahgrid-even}--\eqref{racahgrid-odd} degenerates to the linear one:
\begin{equation}
\lim_{\theta\to\infty}
\frac{-2\left(x^2_{2s}+a(\theta)^2\right)}{\theta}
=2s,
\qquad
\lim_{\theta\to\infty}
\frac{-2\left(x^2_{2s+1}+a(\theta)^2\right)}{\theta}
=2s+\gamma.
\end{equation}
Introducing $P_n(y)$ related to $R_n(x^2; N, a, c, \sigma)$
through
\begin{equation}
x^2
=
-\frac{\theta}{2}y
-
\left(
\frac{\theta-\frac{\gamma}{2}}{2}
\right)^2,
\qquad
P_n(y)
=
\left(-\frac{\theta}{2}\right)^{-n}
R_n\!\left(
x^2;
N,
\frac{\theta-\frac{\gamma}{2}}{2},
\frac{\theta+\frac{\gamma}{2}}{2},
\frac{1}{2}
\right),
\end{equation}
the recurrence relation becomes
\begin{equation}
yP_n(y)
=
P_{n+1}(y)
-
\frac{2b_n+2a(\theta)^2}{\theta}\,P_n(y)
+
\frac{4u_n}{\theta^2}\,P_{n-1}(y).
\label{eq:P-recurrence}
\end{equation}
With \eqref{eq:bn-pararacah} and \eqref{eq:un-pararacah}, the limit $\theta\to\infty$ gives, for $N$ odd,
\begin{align}
&\lim_{\theta\to\infty}
-\frac{2b_n+2a(\theta)^2}{\theta}
=
-\frac{(2j+p-n)(2j+2p-2-2n+\gamma)}{2(2n-2j-2p+1)}-\frac{n(2j+2-2n-\gamma)}{2(2n-2j-1)}\,,
\\
&\lim_{\theta\to\infty}
\frac{4u_n}{\theta^2}
=
\frac{n\left(2 j +p -n +1\right) \left(2 j +2 p -2 n +\gamma \right)  \left(2 j +2-2 n -\gamma \right)}{4 \left(2 n -1-2 j -2 p \right) \left(2 n -2 j -1\right)}\,.
\end{align}
These are the para-Krawtchouk recurrence coefficients \eqref{recurrenceB}, \eqref{recurrenceU} for
the two parities. The para-Krawtchouk polynomials are thus a limiting case
of the para-Racah polynomials.

\subsubsection{The Racah algebra}
The connection between the para-Racah polynomials and the Racah algebra can be seen directly at the level of their operators. The multiplication operator by the variable and the difference operator associated with the para-Racah polynomials satisfy closed commutation relations involving quadratic combinations of the two operators. These relations are those of the Racah algebra, with the corresponding structure constants taking different values according to the parity of $N$. We give these commutation relations explicitly. Let $X$ be the operator multiplication by $x^2$ and $Y$ the difference operator on the r.h.s of \eqref{Y_para}. The commutation relations read 
\begin{equation}
\begin{aligned}
[Y,[X,Y]] &=
2Y^2
-2\{X,Y\}
+\mu Y
+\zeta X
+\xi, \\[1ex]
[[X,Y],X] &=
2\{X,Y\}
-2X^2
+ Y
+\mu X
+\kappa,
\end{aligned}
\end{equation}
with
\begin{align}
\mu
&=
-2(
a(a+j)+c(c+j+p-1)-j(j+p))
-1+p,
\\
\zeta
&=
-(2j+p-1)(2j+p+1)
=
1-N^2,
\\
\xi
&=
-(2j+p+1)
\left[
(j+p-1)a(a+j)
+jc(c+j+p-1)
\right],
\\
\kappa
&=
(j+p-1)a(a+j)
+jc(c+j+p-1)
-2ac(a+j)(c+j+p-1).\label{racsts}
\end{align}

\subsubsection{The para-Racah decomposition}
\label{sec:pR-decomp}

The decomposition of the module carried by the para-Racah polynomials
proceeds along the same lines as for the para-Krawtchouk case in
Subsection~\ref{decomp: Spectral representation}.  We omit some of the
details, as the argument is identical to the para-Krawtchouk case.

\label{sec:pr-decomposition}
\paragraph{Spectral representation.}
Substituting $x=iz$, the spectral points become
 \begin{align}   
&z_{2s}=s+a,\qquad
s=0,\ldots,j,
\label{racahzgrid-even}
\\
&z_{2s+1}=s+c,
\qquad
s=0,\ldots,j-1+p,
\label{racahzgrid-odd}
\end{align}
and equation \eqref{Y_para} reads
\begin{equation}
-n(N-n)R_n(-z^2)
=
\overline{D}(iz)\,R_n\!\left(-(z+1)^2\right)
-
\bigl(\overline{D}(iz)+D(iz)\bigr)R_n(-z^2)
+
D(iz)\,R_n\!\left(-(z-1)^2\right).
\label{new_diff}
\end{equation}
This is exactly the setting of Proposition \ref{prop:mechanism}: condition (i) is the bi-lattice \eqref{racahzgrid-even}--\eqref{racahzgrid-odd}, condition (ii) follows from \eqref{new_diff} together with the explicit form of the bi-lattice 
 \eqref{racahzgrid-even}--\eqref{racahzgrid-odd}, and condition (iii) holds because $D(iz)$
vanishes at the first two spectral points and $\overline{D}(iz)$ at the last
two.  As in Section~\ref{decomp: Spectral representation}, the module is
therefore reducible, and a permutation matrix displays the two blocks. 
Setting
\begin{align}
R^e_n(x)
&=
R_n\big(-(x+a)^2\big),
\qquad
R_n^o(x)
=
R_n\big(-(x+c)^2\big),
\label{transracah}
\\
D^e(x)
&=
D\big(i(x+a)\big),
\qquad
D^o(x)
=
D\big(i(x+c)\big),
\label{transdiff}
\end{align}
one obtains
\begin{align}
\label{eq:transformed-difference-equations}
n(n-N)R^e_n(x)
&=
D^e(x)R^e_n(x-1)-\left(
D^e(x)+\overline{D^e}(x)
\right)R^e_n(x)
+\overline{D^e}(x)R^e_n(x+1),
\\
n(n-N)R^o_n(x)
&=
D^o(x)R^o_n(x-1)-\left(
D^o(x)+\overline{D^o}(x)
\right)R^o_n(x)
+\overline{D^o}(x)R^o_n(x+1),\label{eq:transformed-difference-equations-2}
\end{align}
both of which are the difference equation of the Racah polynomials given by \eqref{racahdiff}. Comparing with \eqref{racahBD} identifies, for $N=2j+1$,
\begin{align}
\alpha_e=-j-1,
\qquad
\beta_e=-j-1,
\qquad
\delta_e=a-c,
\qquad
\gamma_e=a+c-1,
\label{racahcoeffe}\\
\alpha_o=-j-1,
\qquad
\beta_o=-j-1,
\qquad
\delta_o=c-a,
\qquad
\gamma_o=a+c-1.
\label{racahcoeffo}
\end{align}
In both cases $\alpha+\beta+1=-N$, so that the Racah eigenvalue
$n(n+\alpha+\beta+1)$ of the difference equation \eqref{racahdiff} is $n(n-N)$, matching
\eqref{eq:transformed-difference-equations}--\eqref{eq:transformed-difference-equations-2}. 

For $N=2j$, the argument is unchanged; only the dimensions of the two blocks differ, the
even one of dimension $j+1$ and the odd one of dimension $j$, exactly as in
the para-Krawtchouk case. The parameters of the two Racah submodules are
\begin{align}
&\widetilde{\alpha}_e=-j-1,
\qquad
\widetilde{\beta}_e=-j,
\qquad
\widetilde{\delta}_e=a-c,
\qquad
\widetilde{\gamma}_e=a+c-1,
\label{racahevcoeffe}\\
&\widetilde{\alpha}_o=-j,
\qquad
\widetilde{\beta}_o=-j-1,
\qquad
\widetilde{\delta}_o=c-a,
\qquad
\widetilde{\gamma}_o=a+c-1,
\label{racahevcoeffo}
\end{align}
again with $\alpha+\beta+1=-N$ in both cases.

\paragraph{Recurrence representation.}
Let us consider
\begin{align}
H^e_n=H_n(\lambda(x);\alpha_e,\beta_e,\gamma_e,\delta_e,j),
\qquad
H^o_n=H_n(\lambda(x);\alpha_o,\beta_o,\gamma_o,\delta_o,j),\qquad N=2j+1\,,\\
\widetilde{H}_n^e=H_n(\lambda(x);\widetilde{\alpha}_e,\widetilde{\beta}_e,
\widetilde{\gamma}_e,\widetilde{\delta}_e,j),
\qquad
\widetilde{H}_n^o=H_n(\lambda(x);\widetilde{\alpha}_o,\widetilde{\beta}_o,
\widetilde{\gamma}_o,\widetilde{\delta}_o,j-1),\qquad N=2j\,,
\end{align}
where $H_n$ is the standard Racah
polynomial \eqref{racahhyp}, with the last argument ($j,j-1$) referring to the size of the family, and $\lambda(x)=x(x+\gamma+\delta+1)$. Both families are considered on the grid $y_s=\lambda(s)$, $s=0,1,\ldots,j-1+p$. We write  $r_n(x)=\eta_n^{-1}R_n(x)$ for $p=1$ for the non-normalized polynomials. The Whipple transformation formula (see relation (2.10.5) in \cite{GasperRahman2004}) is given by
\begin{equation}
{}_4F_3\left(
\begin{matrix}
-n,\ a,\ b,\ c\\
d,\ e,\ f
\end{matrix}
;1
\right)
=
\frac{(e-a)_n(f-a)_n}{(e)_n(f)_n}
{}_4F_3\left(
\begin{matrix}
-n,\ a,\ d-b,\ d-c\\
d,\ a-e-n+1,\ a-f-n+1
\end{matrix}
;1
\right).
\label{whipple}
\end{equation}
As in Proposition \ref{eq:sim-antisim-parakraw} for the para-Krawtchouk case, three consecutive applications of this transformation yield, for $N$ odd
\begin{align}
   & H^e_n(y_s)=(1-\sigma)r_n(x_{2s})+\sigma r_{N-n}(x_{2s}),
    \label{subevenracah}\\
&H^o_n(y_{s})=\frac{\sigma(a-c-j)_n}{(c-a-j)_n}
    \bigl(r_n(x_{2s+1})-r_{N-n}(x_{2s+1})\bigr),
    \label{suboddracah}
\end{align}
and for $N$ even
\begin{align}
   & \widetilde{H}^e_n(y_{s})=(1-\sigma)\widetilde{r}_n(x_{2s})
    +\sigma\widetilde{r}_{N-n}(x_{2s}),\label{subevenracah-Neven}\\
   & \widetilde{H}^o_n(y_s)=\sigma \frac{j(a-c+1-j)_n}{(j-n)(c-a-j)_n}
    \bigl(\widetilde{r}_n(x_{2s+1})-\widetilde{r}_{N-n}(x_{2s+1})\bigr)\label{suboddracah-Neven},
\end{align}
where $\widetilde{r}_n(x)=\eta_n^{-1}R_n(x)$ for $p=0$. Note that $\widetilde{H}^o_n(y_s)$ is only defined till degree $j-1$, thus avoiding the singularity in \eqref{suboddracah-Neven}. As explained in the remark following
Proposition~\ref{prop:mechanism}, the truncation condition already forces
$\lambda_n=\lambda_{N-n}$ for the para-Racah spectrum, so $r_n$ and $r_{N-n}$ (equivalently for $\widetilde{r}_n$)
share an eigenspace of $Y$ and may be combined within it. As in the para-Krawtchouk case, we have the converse relations for $n\le j$
\begin{equation}
\frac{1-\sigma}{\sigma}r_n(x_{2s+1})+r_{N-n}(x_{2s+1})=0,
\qquad
r_n(x_{2s})-r_{N-n}(x_{2s})=0,
\label{racahparityvanish}
\end{equation}
\begin{equation}
\frac{1-\sigma}{\sigma}
\widetilde{r}_n(x_{2s+1})
+\widetilde{r}_{N-n}(x_{2s+1})=0,
\qquad
\widetilde{r}_n(x_{2s})
-\widetilde{r}_{N-n}(x_{2s})=0.
\end{equation}
If $n > j$, the coefficients pertaining to the odd spectral points equations are swapped between both polynomials. The reduction in the recurrence representation follows as in
Subsection~\ref{Recrep}: \eqref{subevenracah}--\eqref{suboddracah} play the
role of \eqref{evenhahn}--\eqref{oddhahn} (and, for $N$ even, \eqref{subevenracah-Neven}--\eqref{suboddracah-Neven}) play the role of \eqref{evenhahn}--\eqref{oddhahn}.
Inverting the two relations gives, for $N=2j+1$
\begin{equation}
r_n(x_{2s})=
\begin{cases}
H_n^e(y_s), & n\le j,\\[0.8em]
H_{N-n}^e(y_s), & n>j,
\end{cases}
\qquad
r_n(x_{2s+1})=
\begin{cases}
\dfrac{\left(c-a-j\right)_n}
{\left(a-c-j\right)_n}
H_n^o(y_s), & n\le j,\\[1.8em]
\dfrac{(\sigma-1)\left(c-a-j\right)_{N-n}}
{\sigma\left(a-c-j\right)_{N-n}}
H_{N-n}^o(y_s), & n>j.
\end{cases}
\end{equation}
For $N=2j$,
\begin{align}
&\widetilde{r}_n(x_{2s})=
\begin{cases}
\widetilde{H}_n^e(y_s), & n\le j,\\[0.8em]
\widetilde{H}_{N-n}^e(y_s), & n>j,
\end{cases}
\qquad
\widetilde{r}_n(x_{2s+1})=
\begin{cases}
\dfrac{(j-n)\left(c-a-j\right)_n}
{j\left(a-c+1-j\right)_n}
\widetilde{H}_n^o(y_s), & n\le j,\\[1.8em]
\dfrac{(n-j)(\sigma-1)\left(c-a-j\right)_{N-n}}
{j\sigma \left(a-c+1-j\right)_{N-n}}
\widetilde{H}_{N-n}^o(y_s), & n>j.
\end{cases}
\end{align}

\subsubsection{Irreducibility of the submodules}
\label{irredPR}

Here the situation differs from the para-Krawtchouk case, where the
irreducibility came for free.  The Racah parameters \eqref{racahcoeffe}, \eqref{racahcoeffo}, \eqref{racahevcoeffe} and \eqref{racahevcoeffo} are negative integers, and the coefficients
\eqref{racahBD} may vanish inside a sub-lattice.  Evaluating $D(iz)$ and
$\overline{D}(iz)$ on the spectrum, a zero occurs in addition to those at the
two ends when
\begin{equation}
\begin{aligned}
s \in \bigl\{&
a-c,\,
c-a,\,
j+a-c,\,
j+p-1+c-a,\,
-j-2a,\,
-j-a-c,
\\
&
1-p-j-a-c,\,
1-p-j-2c,\,
-2a,\,
-a-c,\,
-2c
\bigr\}.
\end{aligned}
\end{equation}
The positivity of the recurrence coefficients $u_n>0$ \eqref{conditions}, rules out nearly all values. For $N$ even, the remaining possible values are
\begin{equation}
s\in
\left\{
a-c,\,
-j-2a,\,
-j-a-c,\,
1-j-2c,\,
-2a
\right\},
\end{equation}

whereas for $N$ odd, we have
\begin{equation}
s\in
\left\{
-j-2a,\,
-j-2c,\,
-2a,\,
-2c
\right\}.
\end{equation}
However, in all listed cases, the parametrizations necessarily have to be equal to either $j$ or $0$, which are already identified roots of $D(iz)$ and $\overline{D}(iz)$. We can therefore conclude that there are no other spectral points at which the coefficients \eqref{racahBD} vanish, assuring the irreducibility of the submodules.
The para-Krawtchouk case is
recovered in the limit of Subsection \ref{PRsec}, where $a$ and $c$ grow
without bound and all conditions are met automatically; this is the limiting form explained in Section \ref{irred}

\begin{theorem}
\label{thm:PR}
Let $N = 2j+p$ with $p=0,1$, and let $a$, $c$ and $\sigma$ satisfy
\eqref{conditions}.  The $(N+1)$-dimensional module
of the Racah algebra carried by the para-Racah polynomials $V_{N+1}^{PR}$is the direct sum
of two irreducible Racah modules,
\begin{equation}
  V_{N+1}^{\mathrm{PR}} \cong H(\alpha_e, \beta_e,\gamma_e,\delta_e, j)
  \oplus H(\alpha_o, \beta_o,\gamma_o,\delta_o, j),   \qquad p=1,
\end{equation}
\begin{equation}
  V_{N+1}^{\mathrm{PR}} \cong
  H(\widetilde{\alpha}_e, \widetilde{\beta}_e,\widetilde{\gamma}_e,
  \widetilde{\delta}_e, j) \oplus
  H(\widetilde{\alpha}_o, \widetilde{\beta}_o,\widetilde{\gamma}_o,
  \widetilde{\delta}_o, j-1),   \qquad p=0,
\end{equation}
with the parameters given by \eqref{racahcoeffe}, \eqref{racahcoeffo},
\eqref{racahevcoeffe} and \eqref{racahevcoeffo}.  Each summand is irreducible if and only if non of $-j-a-c$, $-2a$ and $-2c$ is an integer lying in $\{1,2,\ldots,j\}$.
\end{theorem}

\subsection{The $q$-para-Racah polynomials}
\label{qPRsec}

The $q$-para-Racah polynomials $\Pi_n(x;a,c,\sigma,N)$ \cite{qPR} arise from
a non-conventional truncation of the Askey--Wilson polynomials. We review the procedure and collect the properties of the polynomials it produces. The Askey--Wilson polynomials, which depend on four parameters $a,b,c,d$, are recalled in Appendix \ref{app:askeywilson}
\paragraph{Truncation condition.}
The Askey--Wilson polynomials $W_n(x;a,b,c,d\,|\,q)$ defined in  \eqref{awdef}
have recurrence coefficients $A_n,C_n$ given by \eqref{awA}--\eqref{awC}. The
recurrence relation \eqref{awrec} terminates after a finite number of steps when
\begin{equation}
    A_NC_{N+1}=0
\end{equation}
Besides the standard truncation, which leads to the $q$-Racah polynomials, the family also
truncates to a finite set of polynomials under this non-conventional
condition
\begin{equation}
1-abcdq^{N-1}=0,\qquad N\in\mathbb{Z}_{>0}.
\label{qracahtrunc}
\end{equation}
It is this choice of parameters that gives rise to the
$q$-para-Racah polynomials.
\paragraph{Parametrization.}
Let $N=2j+1$. The condition \eqref{qracahtrunc} is
implemented by
\begin{equation}
b=a^{-1}q^{-j+e_1t},
\quad
d=c^{-1}q^{-j+e_2t},
\quad t\to0\,,
\label{qracahparam-odd}
\end{equation}
where $e_1,e_2$ are deformation parameters, not independent, leading to a
single deformation parameter $\sigma$,
\begin{equation}
\sigma=\frac{e_1}{e_1+e_2},
\qquad
1-\sigma=\frac{e_2}{e_1+e_2}.
\label{qsigmadef-odd}
\end{equation}
For $N=2j$, the condition \eqref{qracahtrunc} is implemented instead by
\begin{equation}
b=a^{-1}q^{-j+e_1t},
\quad
d=c^{-1}q^{-j+1+e_2t},
\quad t\to0\,,
\label{qracahparam-even}
\end{equation}
with $\sigma$ defined as in \eqref{qsigmadef-odd}. The deformation parameter
is written $\sigma$ here and $\alpha$ in \cite{qPR}.

\paragraph{Recurrence relation.}

Substituting \eqref{qracahparam-odd}--\eqref{qracahparam-even} into the recurrence coefficients of the Askey--Wilson
polynomials \eqref{awA}--\eqref{awC}, and letting $t\to0$, we obtain the recurrence relation that defines the
$q$-para-Racah polynomials
\begin{equation}
x\Pi_n(x)=\Pi_{n+1}(x)+b_n\Pi_n(x)+u_n\Pi_{n-1}(x),
\qquad n=0,1,\ldots,N,
\end{equation}
with $\Pi_{-1}(x)=0$ and $\Pi_0(x)=1$. For $N=2j+1$,
\begin{equation}
b_n=
\begin{cases}
\dfrac{(a+c)(q^{j+1}+1)q^{n}(acq^{j}+1)}{2ac(q^{j}+q^{n})(q^{j+1}+q^{n})},
& n\neq j,\,j+1,\\[10pt]
\dfrac{a+a^{-1}}{2}
+\dfrac{\sigma(c-a)(q^{j+1}-1)q^{-j}(acq^{j}-1)}{2ac(q-1)}
-\dfrac{(q^{j}-1)q^{-j}(c-aq)(acq^{j+1}-1)}{2ac(q^{2}-1)},
& n=j,\\[10pt]
\dfrac{a+a^{-1}}{2}
+\dfrac{(1-\sigma)(c-a)(q^{j+1}-1)q^{-j}(acq^{j}-1)}{2ac(q-1)}
-\dfrac{(q^{j}-1)q^{-j}(c-aq)(acq^{j+1}-1)}{2ac(q^{2}-1)},
& n=j+1,
\end{cases}
\label{eq:bn-qpararacah-odd}
\end{equation}
\begin{equation}
u_n=
\begin{cases}
\dfrac{
(q^{n}-1)(q^{n}-q^{2j+2})(acq^{n}-q)(q^{n}-acq^{2j+1})(aq^{n}-cq^{j+1})(cq^{n}-aq^{j+1})
}{
4a^2c^2(q^{j+1}+q^{n})^2(q^{2n}-q^{2j+1})(q^{2n}-q^{2j+3})
}, & n\neq j+1,\\[10pt]
\dfrac{(1-\sigma)\sigma(c-a)^2q^{-2j}(q^{j+1}-1)^2(acq^{j}-1)^2}{4a^2c^2(q-1)^2},
& n=j+1.
\end{cases}
\label{eq:un-qpararacah-odd}
\end{equation}
For $N=2j$,
\begin{equation}
b_n=
\dfrac{a+a^{-1}}{2}
+\dfrac{(q^{n}-1)(acq^{2j}-q^{n})(aq^{j+1}-cq^{n})}{2ac(q^{j}+q^{n})(q^{2j+1}-q^{2n})}
+\dfrac{(q^{2j}-q^{n})(acq^{n}-1)(cq^{j}-aq^{n+1})}{2ac(q^{j}+q^{n})(q^{2j}-q^{2n+1})},
\qquad \text{for all } n,
\label{eq:bn-qpararacah-even}
\end{equation}
\begin{equation}
u_n=
\begin{cases}
\dfrac{
(q^{n}-1)(q^{n}-q^{2j+1})(acq^{n}-q)(q^{n}-acq^{2j})(aq^{n}-cq^{j})(cq^{n}-aq^{j+1})
}{
4a^2c^2(q^{j}+q^{n})(q^{j+1}+q^{n})(q^{2j+1}-q^{2n})^2
}, & n\neq j,\,j+1,\\[10pt]
\dfrac{(1-\sigma)(c-a)q^{-2j}(q^{j}-1)(q^{j+1}-1)(aq-c)(acq^{j}-1)(acq^{j}-q)}{4a^2c^2(q-1)^2(q+1)},
& n=j,\\[10pt]
\dfrac{\sigma(c-a)q^{-2j}(q^{j}-1)(q^{j+1}-1)(aq-c)(acq^{j}-1)(acq^{j}-q)}{4a^2c^2(q-1)^2(q+1)},
& n=j+1.
\end{cases}
\label{eq:un-qpararacah-even}
\end{equation}
 The positivity of the recurrence coefficients 
$u_n>0$ for all $n=1,\ldots,N$, restrict $a$, $c$ and $\sigma$ to
\begin{equation}
0<q<1,\qquad 0<\sigma<1,\qquad c\neq a,\qquad q<\frac{a}{c}<q^{-1},
\qquad ac<1\ \text{ or }\ ac>q^{1-N}.
\label{qconditions-odd}
\end{equation}
\paragraph{Discrete orthogonality.}
Let $N=2j+p$. The $q$-para-Racah polynomials satisfy the discrete
orthogonality
\begin{equation}
\sum_{s=0}^{N}w_s\Pi_n(x_s)\Pi_m(x_s)=h_n\delta_{nm}\,,
\qquad h_n=u_1u_2\cdots u_n\,,
\end{equation}
where $x_s$ is the $q$-quadratic bi-lattice
 \begin{align}   
&x_{2s}
=
\frac{1}{2}\left(a^{-1}q^{-s}+aq^s\right),\qquad
s=0,\ldots,j,\label{eq:bilattice-para-qracah-even}
\\
&x_{2s+1}
=
\frac{1}{2}\left(c^{-1}q^{-s}+cq^s\right),
\qquad
s=0,\ldots,j-1+p\,.
\label{eq:bilattice-para-qracah-odd}
\end{align}
Conditions \eqref{qconditions-odd} are enough to guarantee no collisions between the two sublattices.
Explicit expressions for the weights are given in
\cite[eqs.~(3.33)--(3.35),(4.19)--(4.21)]{qPR}.
 \paragraph{Difference equation.}
They satisfy the $q$-difference equation
\begin{equation}
\label{eq:diff-qpararacah}
q^{-n}(1-q^n)(1-q^{n-N})\Pi_n(x)
=
A(\theta)T_{+}\Pi_n(x)
-
\left[A(\theta)+\overline{A}(\theta)\right]\Pi_n(x)
+
\overline{A}(\theta)T_{-}\Pi_n(x),
\end{equation}
where $x=\cos\theta$, $T_\pm$ act on $e^{\pm i\theta}$ by $T_\pm e^{i\theta}=q^\pm e^{i\theta}$,
$T_\pm e^{-i\theta}=q^{\mp}e^{-i\theta}$, and
\begin{equation}
A(\theta)
=
\frac{
(1-ae^{i\theta})
(1-a^{-1}q^{-j}e^{i\theta})
(1-ce^{i\theta})
(1-c^{-1}q^{-j+1-p}e^{i\theta})
}{
(1-e^{2i\theta})(1-qe^{2i\theta})
},
\qquad
N=2j+p,
\label{qA}
\end{equation}
with $\overline{A}(\theta)$ the complex conjugate.
\paragraph{Hypergeometric representation.}
Their explicit expressions, as in the previous cases, splits according to the degree. If $n\leq j-p$,
\begin{equation}
\Pi_n(x)=\tau_n\,{}_4\phi_3
\left(
\begin{matrix}
q^{-n},\, q^{n-2j-p},\, ae^{i\theta},\, ae^{-i\theta}\\
q^{-j},\, ac,\, ac^{-1}q^{-j+1-p}
\end{matrix}
\,;\,q;q
\right).
\end{equation}

\noindent If $p=1$ and $n=j$,
\begin{equation}
\Pi_j(x)=\tau_j\sum_{k=0}^{j}
\frac{
(q^{-j-1},ae^{i\theta},ae^{-i\theta};q)_k q^k
}{
(q,ac,ac^{-1}q^{-j};q)_k
}.
\end{equation}

\noindent If $p=1$ and $n=j+1$,
\begin{equation}
\begin{aligned}
\Pi_{j+1}(x)
={}&\tau_{j+1}\sum_{k=0}^{j}
\frac{
(q^{-j-1},ae^{i\theta},ae^{-i\theta};q)_k q^k
}{
(q,ac,ac^{-1}q^{-j};q)_k
}\\
&+\tau_{j+1}
\frac{
(q^{-j-1},ae^{i\theta},ae^{-i\theta};q)_{j+1}
q^{j+1}
}{
\sigma(q,ac,ac^{-1}q^{-j};q)_{j+1}
}.
\end{aligned}
\end{equation}

\noindent If $j+1+p\leq n\leq N$,
\begin{equation}
\begin{aligned}
\Pi_n(x)
={}&\tau_n\,{}_4\phi_3
\left(
\begin{matrix}
q^{-n},\, q^{n-2j-p},\,
ae^{i\theta},\, ae^{-i\theta}\\
q^{-j},\, ac,\, ac^{-1}q^{-j+1-p}
\end{matrix}
\,;\,q;q
\right)\\[6pt]
&+\tau_n
\frac{
(q^{n-2j-p};q)_{2j+p-n}
(q^{-n},ae^{i\theta},ae^{-i\theta};q)_{j+1}
(q;q)_{n-j-p}q^{j+1}
}{
\sigma(q^{-j};q)_j
(q,ac,ac^{-1}q^{-j+1-p};q)_{j+1}
}\\[6pt]
&\qquad\times{}_4\phi_3
\left(
\begin{matrix}
q^{j+1-n},\, q^{n-j+1-p},\,
aq^{j+1}e^{i\theta},\,aq^{j+1}e^{-i\theta}\\
q^{j+2},\, acq^{j+1},\, ac^{-1}q^{2-p}
\end{matrix}
\,;\,q;q
\right).
\end{aligned}
\end{equation}
where $\tau_n$ is a normalization factor given in \cite{qPR} that makes the polynomials monic.

\subsubsection{The $q$-Racah algebra}
The connection between the $q$-para-Racah polynomials and the $q$-Racah algebra can be seen directly at the level of their operators.  The multiplication operator by the variable and the difference operator associated with the $q$-para-Racah polynomials satisfy the relations of the the $q$-Racah (Askey Wilson) algebra  (Definition (3.1) in \cite{HAW}). We give these commutation relations explicitly.

Let $X$ be the multiplication by $x$ and $Y$ the operator on the r.h.s of \eqref{eq:diff-qpararacah}. The commutation relations read
\begin{equation}
\begin{aligned}
{}[X,[X,Y]] &= \rho\, XYX + a_1X^2 + a_2\{X,Y\} + a_3X + a_4Y + a_5 I,\\[1ex]
[Y,[Y,X]] &= \rho\, YXY + a_1\{X,Y\} + a_2Y^2 + a_3Y + a_6X + a_7 I,
\end{aligned}
\label{AWalgebra}
\end{equation}
with $\rho=q^2+q^{-2}-2$ and scalars $a_1,\ldots,a_7$ fixed by the
parameters. This is the $q$-Racah, or Askey--Wilson, algebra.  As $q\to1$
one has $\rho\to0$ and \eqref{AWalgebra} becomes the quadratic Racah algebra
\cite{mutual}, in agreement with the limit of Subsection \ref{PRsec}.

\subsubsection{The $q$-para-Racah decomposition}
\label{sec:qpR-decomp}

The decomposition of the module carried by the para-Racah polynomials
proceeds along the same lines as for the para-Krawtchouk and para-Racah cases in
Subsections~\ref{decomp: Spectral representation} and \ref{sec:pR-decomp} .  We omit some of the
details, as the argument is identical to the para-Krawtchouk case.

\paragraph{Spectral representation.}
With $\mu(x)=q^{-x}+\gamma\delta q^{x+1}$, the grid \eqref{eq:bilattice-para-qracah-even}--\eqref{eq:bilattice-para-qracah-odd} gives $\gamma_e\delta_e=a^2q^{-1}$ and $\gamma_o\delta_o=c^2q^{-1}$, so that
\begin{equation}
\mu(s)=q^{-s}+a^2q^{s}=2a\,x_{2s}
\qquad\text{and}\qquad
\mu(s)=q^{-s}+c^2q^{s}=2c\,x_{2s+1}\,,
\label{qrescale}
\end{equation}
on the even and odd sub-lattices respectively: each is the $q$-Racah grid dilated, by $2a$ on one and by $2c$ on the other. This is the counterpart here of the substitution \eqref{eq:transformed-polynomials-a} of Section \ref{reducibility}. This is exactly the setting of Proposition \ref{prop:mechanism}: condition (i) is the bi-lattice \eqref{eq:bilattice-para-qracah-even}--\eqref{eq:bilattice-para-qracah-odd}, condition (ii) follows from the $q$-difference equation \eqref{eq:diff-qpararacah} together with \eqref{qrescale}, and condition (iii) holds because $A(\theta)$ vanishes at the last two spectral points and $\overline{A}(\theta)$ at the first two, for either parity of $N$. 

As in Sections \ref{decomp: Spectral representation} and \ref{sec:pr-decomposition}, the module is therefore reducible, and a permutation matrix displays the two blocks.
Upon comparing the coefficients of the $q$-difference equation \eqref{eq:diff-qpararacah}, evaluated at the spectral points, to those of the $q$-Racah difference equation \eqref{qracahdiff}, we identify the parameters of both submodules. For $N=2j+1$,
\begin{equation}
\alpha_e=q^{-j-1},
\qquad
\beta_e=q^{-j-1},
\qquad
\delta_e=ac^{-1},
\qquad
\gamma_e=acq^{-1},
\label{qracahcoeffe}
\end{equation}
\begin{equation}
\alpha_o=q^{-j-1},
\qquad
\beta_o=q^{-j-1},
\qquad
\delta_o=a^{-1}c,
\qquad
\gamma_o=acq^{-1},
\label{qracahcoeffo}
\end{equation}
and for $N=2j$,
\begin{equation}
\widetilde{\alpha}_e=q^{-j-1},
\qquad
\widetilde{\beta}_e=q^{-j},
\qquad
\widetilde{\delta}_e=ac^{-1},
\qquad
\widetilde{\gamma}_e=acq^{-1},
\label{qracahevcoeffe}
\end{equation}
\begin{equation}
\widetilde{\alpha}_o=q^{-j},
\qquad
\widetilde{\beta}_o=q^{-j-1},
\qquad
\widetilde{\delta}_o=a^{-1}c,
\qquad
\widetilde{\gamma}_o=acq^{-1}.
\label{qracahevcoeffo}
\end{equation}
In all cases $q\alpha\beta=q^{-N},$ so that the $q$-Racah eigenvalue is $q^{-n}(1-q^n)(1-q^{n-N})$, matching \eqref{eq:diff-qpararacah}.

\paragraph{Recurrence representation.}
Let 
\begin{align}
&\Xi^e_n=\Xi_n(\mu(x);\alpha_e,\beta_e,\gamma_e,\delta_e,j), \qquad \Xi^o_n=\Xi_n(\mu(x);\alpha_o,\beta_o,\gamma_o,\delta_o,j), \qquad \text{$N=2j+1$,}\\
&\Xi^e_n=\Xi_n(\mu(x);\widetilde{\alpha}_e,\widetilde{\beta}_e,\widetilde{\gamma}_e,\widetilde{\delta}_e,j), \qquad \Xi^o_n=\Xi_n(\mu(x);\widetilde{\alpha}_o,\widetilde{\beta}_o,\widetilde{\gamma}_o,\widetilde{\delta}_o,j-1), \qquad \text{$N=2j$,}
\end{align}
where $\Xi_n(\mu(x);\alpha,\beta,\gamma,\delta\mid q)$ is the standard
$q$-Racah polynomial \eqref{eq:qracah-defi} with $
\mu(x)=q^{-x}+\gamma\delta q^{x+1}$. Both families are considered on the grid $\zeta_s=\mu(s),\, s=0,1,\ldots,j$.  The Sears transformation formula (see relation (III.15) in \cite{GasperRahman2004}) is given by
\begin{equation}
{}_4\phi_3\left(
\begin{matrix}
q^{-n},\, a,\, b,\, c\\
d,\, e,\, f
\end{matrix}
; q, q
\right)
=
\frac{(a^{-1}e,a^{-1}f;q)_n}{(e,f;q)_n}\,a^n
{}_4\phi_3\left(
\begin{matrix}
q^{-n},\, a,\, b^{-1}d,\, c^{-1}d\\
d,\, ae^{-1}q^{1-n},\, af^{-1}q^{1-n}
\end{matrix}
; q, q
\right).
\end{equation}
Three consecutive applications of this transformation yield, for
$N$ odd,
\begin{align}
    &\Xi^e_n(\zeta_s)=(1-\sigma)\pi_n(x_{2s})+\sigma \pi_{N-n}(x_{2s}),
    \label{subevenqracah}\\
    &\Xi^o_n(\zeta_s)=\frac{\sigma a^n(a^{-1}cq^{j};q^{-1})_n}
    {c^n(ac^{-1}q^j;q^{-1})_n}
    \bigl(\pi_n(x_{2s+1})-\pi_{N-n}(x_{2s+1})\bigr),
    \label{suboddqracah}
\end{align}
and, for $N$ even,
\begin{align}
&\widetilde{\Xi}^e_n(\zeta_s)=(1-\sigma)\widetilde{\pi}_n(x_{2s})
    +\alpha \widetilde{\pi}_{N-n}(x_{2s}),\label{subevenqracahN}\\
&\widetilde{\Xi}^o_n(\zeta_s)=
    \frac{\sigma a^n(a^{-1}cq^{j-1};q^{-1})_n[j]_q}   {c^n(ac^{-1}q^j;q^{-1})_n[j-n]_q}
    \bigl(\widetilde{\pi}_n(x_{2s+1})-\widetilde{\pi}_{N-n}(x_{2s+1})\bigr)\label{suboddqracahN},
\end{align}
where $[a]_q=(q^a;q)_1/(q;q)_1$ and $\pi_n$ and $\widetilde{\pi}_n$ are the non-normalized
$q$-para-Racah polynomial for $N$ odd and even respectively. The reduction in the recurrence representation follows as in Subsection~\ref{Recrep}: \eqref{subevenqracah}--\eqref{suboddqracah} (and, for $N$ even, \eqref{subevenqracahN}--\eqref{suboddqracahN}) play the role of \eqref{evenhahn}--\eqref{oddhahn}.
Inverting these relations expresses the
$q$-para-Racah polynomials through the polynomials of the two submodules, as
before. For $N$ odd, we have
\begin{align}
&\pi_n(x_{2s})=
\begin{cases}
\Xi_n^e(\zeta_s), & \text{if } n\le j,\\[0.8em]
\Xi_{N-n}^e(\zeta_s), & \text{if } n>j.
\end{cases}\\
&\pi_n(x_{2s+1})=
\begin{cases}
\frac{c^n(ac^{-1}q^j;q^{-1})_n}{a^n(a^{-1}cq^j;q^{-1})_n}\Xi_n^o(\zeta_s), & \text{if } n\le j,\\[0.8em]
\frac{(\alpha-1)c^n(ac^{-1}q^j;q^{-1})_{N-n}}{\alpha a^n(a^{-1}cq^j;q^{-1})_{N-n}}\Xi_{N-n}^o(\zeta_s), & \text{if } n>j.
\end{cases}
\end{align}
When $N$ is even, we have
\begin{align}
&\widetilde{\pi}_n(x_{2s})=
\begin{cases}
\widetilde{\Xi}_n^e(\zeta_s), & \text{if } n\le j,\\[0.8em]
\widetilde{\Xi}_{N-n}^e(\zeta_s), & \text{if } n>j.
\end{cases}\\
&\widetilde{\pi}_n(x_{2s+1})=
\begin{cases}
\frac{c^n(ac^{-1}q^j;q^{-1})_n[j-n]_q}{a^n(a^{-1}cq^{j-1};q^{-1})_n[j]_q}\widetilde{\Xi}_n^o(\zeta_s), & \text{if } n\le j,\\[0.8em]
\frac{(\sigma-1)c^n(ac^{-1}q^j;q^{-1})_{N-n}[n-j]_q}{\sigma a^n(a^{-1}cq^{j-1};q^{-1})_{N-n}[j]_q}\widetilde{\Xi}_{N-n}^o(\zeta_s), & \text{if } n>j.
\end{cases}
\end{align}
\subsubsection{Irreducibility of the submodules}
\label{irredqPR}
As in the case of the para-Racah polynomials, we are not guaranteed the irreducibility. The $q$-Racah parameters \eqref{qracahcoeffe}, \eqref{qracahcoeffo}, \eqref{qracahevcoeffe}, and \eqref{qracahevcoeffo} can lead to cancelation of coefficients \eqref{qracahBD}. Evaluating $A(\theta)$ and $\overline{A}(\theta)$ on the spectrum, a zero occurs in addition to those at the two ends when
\begin{equation}
\begin{aligned}
s \in \biggl\{&
\log_q(a^{-2}),\,
\log_q\!\left((ac)^{-1}\right),\,
j-1+p+\log_q\!\left(\frac{c}{a}\right),\,
-j+\log_q(a^{-2}),
\\
&
\log_q\!\left(\frac{c}{a}\right),\,
1-p-j+\log_q\!\left((ac)^{-1}\right),\,
j+\log_q\!\left(\frac{a}{c}\right),
\\
&
\log_q(c^{-2}),\,
\log_q\!\left(\frac{a}{c}\right),\,
-j+\log_q\!\left((ac)^{-1}\right),\,
1-p-j+\log_q(c^{-2})
\biggr\}.
\end{aligned}
\end{equation}
However, the conditions required for the positivity of the recurrence coefficients $u_n>0$ \eqref{qconditions-odd} rule out all of the listed values, assuring that the expressions for $A(\theta)$ and $\overline{A}(\theta)$ only give 0 for the already identified spectral points. 

\begin{theorem}
\label{thm:qPR}
Let $N=2j+p$ with $p=0,1$, and let $a,c,\sigma,q$ satisfy \eqref{qconditions-odd}. The $(N+1)$-dimensional module of the $q$-Racah algebra carried by the $q$-para-Racah polynomials $V_{N+1}^{\mathrm{qPR}}$ is the direct sum of two irreducible $q$-Racah modules,
\begin{equation}
V_{N+1}^{\mathrm{qPR}} \cong \Xi(\alpha_e,\beta_e,\gamma_e,\delta_e,j)\oplus\Xi(\alpha_o,\beta_o,\gamma_o,\delta_o,j),
\qquad p=1,
\end{equation}
\begin{equation}
V_{N+1}^{\mathrm{qPR}} \cong \Xi(\widetilde{\alpha}_e,\widetilde{\beta}_e,\widetilde{\gamma}_e,\widetilde{\delta}_e,j)\oplus\Xi(\widetilde{\alpha}_o,\widetilde{\beta}_o,\widetilde{\gamma}_o,\widetilde{\delta}_o,j-1),
\qquad p=0,
\end{equation}
with the parameters given by \eqref{qracahcoeffe}, \eqref{qracahcoeffo}, \eqref{qracahevcoeffe} and \eqref{qracahevcoeffo}.
\end{theorem}

\subsection{The para-Bannai--Ito polynomials}
The para-Bannai-Ito polynomials $\boldsymbol{P}_n^{(0,1)}(x;a,b,\sigma,N)$ \cite{paraBI} arise through a truncation of the general untruncated Bannai-Ito polynomials (BI) in the odd case, and of the general untruncated complementary Bannai-Ito polynomials (CBI) in the even case. We review the procedure and display the properties of the given polynomials. The BI and CBI, which depend on four parameters $\rho_1$, $\rho_2$,$r_1$,$r_2$ are recalled in \ref{app:BI}.  
\label{PBIsec}
\paragraph{Truncation condition.}
The BI polynomials $B_n(x;\rho_1,\rho_2,r_1,r_2)$  have recurrence coefficients $A_n,C_n$ given by \eqref{BIAn}--\eqref{BICn}. The CBI polynomials   $W_n(x;\rho_1,\rho_1,r_1,r_2)$ have recurrence relation \eqref{CBIrec} with coefficients \eqref{CBItaueven}--\eqref{CBItauodd} . Both recurrence relation terminate after a finite number of steps when
\begin{equation}
    A_NC_{N+1}=0\,.
\end{equation}
Besides the standard truncation, which leads to the usual truncated BI/CBI polynomials, the family also truncates to a finite set of polynomials when
\begin{equation}
  N+2g+1=0 \quad (N \text{ odd}),
  \qquad
  N+2g+2=0 \quad (N \text{ even}),
  \qquad
  g=\rho_1+\rho_2-r_1-r_2,
  \label{PBItruncation}
\end{equation}
applied to the Bannai--Ito polynomials when $N$ is odd and to the
complementary Bannai--Ito polynomials when $N$ is even.  Writing
$N=2j+p$, both conditions read
\begin{equation}
  j+g+1=0,
  \label{PBItrunc2}
\end{equation}
and the parametrization realizing them depends on the parity of $j$, not only
on that of $N$ \cite{paraBI}.  We treat $j$ even below; the case $j$ odd runs
along the same lines with the parametrizations given in \cite{paraBI}. 

\paragraph{Parametrization.}
The condition \eqref{PBItrunc2} can be implemented  by the following parametrization
\begin{equation}
\rho_1-r_1=-\frac{j+1}{2}+e_1t,
\qquad
\rho_2-r_2=-\frac{j+1}{2}+e_2t,
\quad t\to0\,,\label{PBIparam}
\end{equation}
where $e_1,e_2$ are deformation parameters. Using a change of parameters, we obtain 
\begin{equation}
\rho_1=\frac{b-j-1+a}{4},\qquad
\rho_2=\frac{b-j-1-a}{4},\qquad
\frac{e_1}{e_1+e_2}=\sigma,\qquad
\frac{e_2}{e_1+e_2}=1-\sigma.
\end{equation}
\begin{equation}
\sigma=\frac{e_1}{e_1+e_2},
\qquad
1-\sigma=\frac{e_2}{e_1+e_2}.
\label{sigmadef}
\end{equation}
\paragraph{Recurrence relation.}
Substituting \eqref{PBIparam} in the recurrence coefficients, we get the recurrence relation that defines the monic para-Bannai--Ito polynomials
\begin{equation}
xP_n^{(p)}(x)
=
P_{n+1}^{(p)}(x)
+
\left(
\frac{b-j-1+a}{4}-A_n^p-C_n^p
\right)P_n^{(p)}(x)
+
A_{n-1}^p C_n^p P_{n-1}^{(p)}(x),
\end{equation}
with
\begin{equation}
A_n^p=
\begin{cases}
\dfrac{1}{4}\bigl(n-j+(2p-1)a\bigr),
& n\text{ even},\ n\neq j,\\[3mm]
\dfrac{
p(n-2j-1)(n-j+b)
-(1-p)(n+1)(n-j-b)
}{4(n-j)},
& n\text{ odd},\\[3mm]
\dfrac{a}{2}\bigl[p\sigma+(1-p)(1-\sigma)\bigr],
& n\text{ even},\ n=j,
\end{cases}
\label{AnPBI}
\end{equation}

\begin{equation}
C_n^p=
\begin{cases}
\dfrac{1-p}{4}(n-j+a)
-\dfrac{p\,n(n-j-1-b)}{4(n-j-1)},
& n\text{ even},\ n\neq j,\\[3mm]
\dfrac{(1-p)(n-2j-1)(n-j+b)}{4(n-j)}
-\dfrac{p}{4}(n-j-1-a),
& n\text{ odd},\ n\neq j+1,\\[3mm]
-\dfrac{(1-p)j(1+b)}{4}
+\dfrac{p(1-\sigma)a}{2},
& n\text{ odd},\ n=j+1,
\end{cases}
\label{CnPBI}
\end{equation}
with $\boldsymbol{P}_{-1}^{(p)}=0$ and $\boldsymbol{P}_{0}^{(p)}=1$. The positivity of the recurrence coefficients $A_{n-1}^p C_n^p>0$ for all $n=1,\ldots,N$ restricts $a$,$b$ and $\sigma$ to 
\begin{equation}
\begin{cases}
\begin{aligned}
&a\leq -j-1,\quad |b|\leq 1,\quad 0\leq\alpha\leq 1,\\
&\qquad\text{or}\\
&b\geq j,\quad |a+1|\leq 1,\quad 0\leq\alpha\leq 1,
\end{aligned}
& p=0,\\[4mm]
|a|\geq j+1,\quad |b|\leq 1,\quad 0\leq\alpha\leq 1,
& p=1.
\end{cases}
\end{equation}
\paragraph{Discrete orthogonality.} The para-Bannai-Ito polynomials satisfy the discrete orthogonality
\begin{equation}
\sum_{s=0}^{2j+p} w_s P_n^{(p)}(x_s)P_m^{(p)}(x_s)
= h_n\delta_{nm}.
\end{equation}
where $x_s$ is the linear bi-lattice:
\begin{subequations}
 \begin{align}   
&x_{2s}
=
-(-1)^s\left(\frac{2s-j-b+a}{4}\right)-\frac{1}{4},\qquad
s=0,\ldots,j,
\\[8pt]
&x_{2s+1}
=-(-1)^s\left(\frac{2s-j-b-a}{4}\right)-\frac{1}{4},
\qquad
s=0,\ldots,j-1+p.
\end{align}
\label{PBIgrid}
\end{subequations}
Explicit expressions for the weights are given in \cite{paraBI}.To ensure that there are no collisions between the two sublattices, we need to impose one of the following conditions:
\begin{equation}
\left\{
\begin{aligned}
&a\leq -j-1,\\
&|b|< 1,\\
\end{aligned}
\right.
\qquad \text{or} \qquad
\left\{
\begin{aligned}
&b\geq j,\\
&-2<a<0.
\end{aligned}
\right.
\end{equation}
\paragraph{Difference equation}
The polynomials $\boldsymbol{P}_n^{(1)}(x;a,b,\sigma,N)$ satisfy
$L\boldsymbol{P}_n^{(1)}=\lambda_n\boldsymbol{P}_n^{(1)}$ with
$\lambda_{2n}=n$, $\lambda_{2n+1}=j-n$, where
\begin{equation}
L=F(x)(I-R)+G(x)(T^{+}R-I),
\label{PBIdunkl}
\end{equation}
\begin{equation}
G(x)=\frac{(4x+1+a-b-j)(4x+1-a-b-j)}
{16(2x+1)},
\qquad
F(x)=\frac{(4x+1+j+a-b)(4x+1+j-a-b)}{32x},
\end{equation}
and $T^{\pm}f(x)=f(x\pm1)$, $Rf(x)=f(-x)$.

The polynomials $\boldsymbol{P}_n^{(0)}(x;a,b,\sigma,N)$ satisfy
$\mathcal{D}_{\beta}\boldsymbol{P}_n^{(0)}
=\Lambda_n^{(\beta)}\boldsymbol{P}_n^{(0)}$ with
$\Lambda_{2n}^{(\beta)}=n(n-j)$ and
$\Lambda_{2n+1}^{(\beta)}=n(n+1-j)+\beta$, where
\begin{equation}
\label{eq:diff-PBP0}
\mathcal{D}_{\beta}
=\mathcal{D}_0
+\beta\frac{x-\rho_2}{2x}(I-R),
\end{equation}
\begin{equation}
\mathcal{D}_0
=A(x)T^{+}+B(x)T^{-}+C(x)R+D(x)T^{+}R
-\bigl(A(x)+B(x)+C(x)+D(x)\bigr)I,
\end{equation}
and
\begin{equation}
A(x)=\frac{(x+\rho_1+1)(x+\rho_2+1)(2x-2\rho_1-j)(2x-2\rho_2-j)}
{8(x+1)(2x+1)},
\end{equation}
\begin{equation}
B(x)=\frac{(x-\rho_2)(x-\rho_1-1)(2x+2\rho_1+j)(2x+2\rho_2+j)}{8x(2x-1)},
\end{equation}
\begin{equation}
C(x)=\frac{(x-\rho_2)(4x^2+\omega)}{8x}
-\frac{(x-\rho_2)(x+\rho_1+1)(2x-2\rho_1-j)(2x-2\rho_2-j)}{8x(2x+1)}-B(x),
\end{equation}
\begin{equation}
D(x)=\frac{\rho_2(x+\rho_1+1)(2x-2\rho_1-j)(2x-2\rho_2-j)}{8x(x+1)(2x+1)},
\end{equation}
\begin{equation}
\omega=(2\rho_1+j)(2\rho_2+j)-4(1+\rho_1)(\rho_1+\rho_2+j),
\end{equation}
where $\rho_1$ and $\rho_2$ are given by \eqref{PBIparam}.
\paragraph{Hypergeometric representation.}
The PBI polynomials follow the pattern of the para family, having split representations depending on the degree $n$. For $N$ even, we have
\begin{equation}
P_{2n}^{(0)}(x)
=\kappa_n^{(1)}
\,{}_4F_3\left(
\begin{matrix}
-n,\;n-j,\;\frac{b-j-1-a}{4}+x,\;
\frac{b-j-1-a}{4}-x\\
\frac{1+b-j}{2},\;-\frac{j+a}{2},\;-\frac{j}{2}
\end{matrix}
;1
\right),
\end{equation}
if $n\leq \frac{j}{2}$, and
\begin{equation}
\begin{aligned}
P_{2n}^{(0)}(x)
={}&\kappa_n^{(1)}
\,{}_4F_3\left(
\begin{matrix}
-n,\;n-j,\;\frac{b-j-1-a}{4}+x,\;
\frac{b-j-1-a}{4}-x\\
\frac{1+b-j}{2},\;-\frac{j+a}{2},\;-\frac{j}{2}
\end{matrix}
;1
\right)\\[6pt]
&+\kappa_n^{(1)}
\frac{
(n-j)_{j-n}(-n)_{\frac{j}{2}+1}
\left(\frac{b-j-1-a}{4}+x\right)_{\frac{j}{2}+1}
\left(\frac{b-j-1-a}{4}-x\right)_{\frac{j}{2}+1}
(1)_{n-\frac{j}{2}}
}{
(1-\alpha)
\left(-\frac{j}{2}\right)_{\frac{j}{2}}
(1)_{\frac{j}{2}+1}
\left(\frac{1+b-j}{2}\right)_{\frac{j}{2}+1}
\left(-\frac{j+a}{2}\right)_{\frac{j}{2}+1}
}\\[6pt]
&\qquad\times{}_4F_3\left(
\begin{matrix}
\frac{j+2}{2}-n,\;n-\frac{j-2}{2},\;
\frac{b+j+3-a}{4}+x,\;
\frac{b+j+3-a}{4}-x\\
\frac{3+b}{2},\;\frac{2-a}{2},\;\frac{4+j}{2}
\end{matrix}
;1
\right).
\end{aligned}
\end{equation}
if $n>\frac{j}{2}$. For odd degrees, we have
\begin{equation}
\begin{aligned}
P_{2n+1}^{(0)}(x)
={}&\kappa_n^{(2)}
\left(x-\frac{b-j-1-a}{4}\right)\\
&\times{}_4F_3\left(
\begin{matrix}
-n,\;n+1-j,\;
\frac{b-j+3-a}{4}+x,\;
\frac{b-j+3-a}{4}-x\\
\frac{3+b-j}{2},\;\frac{2-j-a}{2},\;
\frac{2-j}{2}
\end{matrix}
;1
\right),
\end{aligned}
\end{equation}
if $n<\frac{j}{2}$, and
\begin{equation}
\begin{aligned}
P_{2n+1}^{(0)}(x)
={}&\kappa_n^{(2)}
\left(x-\frac{b-j-1-a}{4}\right)
\Bigg[
{}_4F_3\left(
\begin{matrix}
-n,\;n+1-j,\;
\frac{b-j+3-a}{4}+x,\;
\frac{b-j+3-a}{4}-x\\
\frac{3+b-j}{2},\;\frac{2-j-a}{2},\;
\frac{2-j}{2}
\end{matrix}
;1
\right)\\[6pt]
&+\frac{
(n+1-j)_{j-n-1}(-n)_{\frac{j}{2}}
\left(\frac{b-j+3-a}{4}+x\right)_{\frac{j}{2}}
\left(\frac{b-j+3-a}{4}-x\right)_{\frac{j}{2}}
(1)_{n-\frac{j}{2}}
}{
(1-\alpha)
\left(\frac{2-j}{2}\right)_{\frac{j}{2}-1}
(1)_{\frac{j}{2}}
\left(\frac{3+b-j}{2}\right)_{\frac{j}{2}}
\left(\frac{2-j-a}{2}\right)_{\frac{j}{2}}
}\\[6pt]
&\qquad\times{}_4F_3\left(
\begin{matrix}
\frac{j}{2}-n,\;n-\frac{j-2}{2},\;
\frac{b+j+3-a}{4}+x,\;
\frac{b+j+3-a}{4}-x\\
\frac{3+b}{2},\;\frac{2-a}{2},\;\frac{4+j}{2}
\end{matrix}
;1
\right)
\Bigg].
\end{aligned}
\end{equation}
if $n\geq\frac{j}{2}$. $\kappa^{(1,2)}$ are normalization factors to ensure that the polynomials are monic. Their explicit expressions are given in \cite{paraBI}.\\
For $N$ odd, we can obtain the hypergeometric representation using the even expression, by the fact that the odd PBI polynomials $\boldsymbol{P}_n^{(1)}(x)$ are obtained by performing a Geronimus transformation on the even PBI polynomials $\boldsymbol{P}_n^{(0)}(x)$. Their explicit expression is thus:
\begin{equation}
P_n^{(1)}(x)=P_n^{(0)}(x)-C_n^{1}P_{n-1}^{(0)}(x),
\end{equation}
where $C_n^1$ is given by \eqref{CnPBI}.
\subsubsection{The para-Bannai--Ito decomposition: $N$ odd case}
The polynomials $\mathbf{P}_n^{(1)}$ satisfy the difference equation \eqref{PBIdunkl}. On the grid \eqref{PBIgrid},
\begin{align}
T^{+}x_{2k+p}
&=
\begin{cases}
x_{2k+p-4}, & k \text{ even},\\
x_{2k+p+4}, & k \text{ odd},
\end{cases}
&
T^{-}x_{2k+p}
&=
\begin{cases}
x_{2k+p+4}, & k \text{ even},\\
x_{2k+p-4}, & k \text{ odd},
\end{cases}
\\[8pt]
Rx_{2k+p}
&=
\begin{cases}
x_{2k+p+2}, & k \text{ even},\\
x_{2k+p-2}, & k \text{ odd},
\end{cases}
&
T^{+}Rx_{2k+p}
&=
\begin{cases}
x_{2k+p-2}, & k \text{ even},\\
x_{2k+p+2}, & k \text{ odd}.
\end{cases}
\end{align}
All four operators shift within a sub-lattice, which is condition (ii);
$G(x)$ vanishes at the first two spectral points and $F(x)$ at the last two,
which is condition (iii).

The operator \eqref{PBIdunkl} is the Bannai--Ito Dunkl shift operator
\eqref{BIdunkl} of Appendix \ref{app:BI} at the parameter values
\begin{equation}
\rho_1=\frac{b-j-1+a}{4},
\quad
\rho_2=\frac{b-j-1-a}{4},
\quad
r_1=\frac{b+j+1+a}{4},
\quad
r_2=\frac{b+j+1-a}{4},
\label{PBIparam}
\end{equation}
as substitution in the expressions for for $F(x)$ and $G(x)$.  Restricting it to each
sub-lattice identifies the parameters of the two submodules:
\begin{equation}
    \rho^e_1=\frac{b-j+a-1}{4}, \qquad
    \rho^e_2=\frac{b-j-a-1}{4}, \qquad
    r^e_1=\frac{j+b-a+1}{4}, \qquad
    r^e_2=\frac{j+b+a+1}{4},
    \label{oBIe}
\end{equation}
\begin{equation}
    \rho^o_1=\frac{b-j+a-1}{4}, \qquad
    \rho^o_2=\frac{b-j-a-1}{4}, \qquad
    r^o_1=\frac{j+b+a+1}{4}, \qquad
    r^o_2=\frac{j+b-a+1}{4}.
    \label{oBIo}
\end{equation}
These are the parameters \eqref{PBIparam} of the parent family itself, the
two submodules differing only by the exchange $r_1\leftrightarrow r_2$.  The
structure constants $\omega_1$, $\omega_2$, $\omega_3$ of the Bannai--Ito
algebra are symmetric in $r_1$ and $r_2$, so that the two summands are
modules of one and the same Bannai--Ito algebra.  This is the counterpart of
what was observed in Subsection \ref{mechanism} for the Hahn case, where the
two summands carry exchanged parameters and a common eigenvalue parameter.
The para-Bannai--Ito polynomials being a $q\to-1$ limit of the
$q$-para-Racah polynomials \cite{paraBI}, obtained through
\begin{equation}
q=-e^{\epsilon},
\quad
a\to i\exp\left(\epsilon\left[\tfrac{a+b-j}{2}\right]\right),
\quad
c \to i\exp\left(\epsilon\left[\tfrac{-a+b-j}{2}\right]\right),
\quad
x\rightarrow 2i\epsilon\left(x+\tfrac{1}{4}\right)
\label{limitpbi}
\end{equation}
and $\epsilon\to0$, the same limit applied to \eqref{subevenqracah} and
\eqref{suboddqracah} yields
\begin{equation}
b_n^e(x_{2s})=(1-\sigma)\boldsymbol{P}_n^{(1)}(x_{2s})
+\alpha\boldsymbol{P}^{(1)}_{N-n}(x_{2s}),
\end{equation}
\begin{equation}
b_n^o(x_{2s+1})=\sigma\frac{\left(\frac{a}{2}\right)_{\left\lfloor n/2\right\rfloor}}
{\left(-\frac{a}{2}\right)_{\left\lfloor n/2\right\rfloor}}
\left[\boldsymbol{P}_n^{(1)}(x_{2s+1})-\boldsymbol{P}^{(1)}_{N-n}(x_{2s+1})
\right],
\end{equation}
with $b_n$ the non-monic Bannai--Ito polynomial and the parameters
\eqref{oBIe}, \eqref{oBIo}.  Inverting,
\begin{equation}
\boldsymbol{P}_n^{(1)}(x_{2s})
=\begin{cases}
b_n^e(x_{2s}), & n\leq j,\\[2mm]
b_{N-n}^e(x_{2s}), & n>j,
\end{cases}
\qquad
\boldsymbol{P}_n^{(1)}(x_{2s+1})
=\begin{cases}
\dfrac{
\left(-\frac{a}{2}\right)_{\left\lfloor n/2\right\rfloor}
}{
\left(\frac{a}{2}\right)_{\left\lfloor n/2\right\rfloor}
}
b_n^o(x_{2s+1}),
& n\leq j,\\[5mm]
\dfrac{\sigma-1}{\sigma}
\dfrac{
\left(-\frac{a}{2}\right)_{\left\lfloor (N-n)/2\right\rfloor}
}{
\left(\frac{a}{2}\right)_{\left\lfloor (N-n)/2\right\rfloor}
}
b_{N-n}^o(x_{2s+1}),
& n>j.
\end{cases}
\end{equation}

\subsubsection{The para-Bannai-Ito decomposition: $N$ even case}

The polynomials $\boldsymbol{P}_n^{(0)}(x;a,b,\sigma,N)$ satisfy the difference equation \eqref{eq:diff-PBP0} 
where $\rho_1$ and $\rho_2$ are given by \eqref{PBIparam}. The grid is \eqref{PBIgrid} with $N+1=2j+1$ points.  Writing
$C'(x)=C(x)-\beta(x-\rho_2)/(2x)$, one finds that $A$ vanishes at
$x_0,x_1,x_{2j-2},x_{2j-1}$, that $B$ vanishes at $x_2,x_3,x_{2j-2},x_{2j}$,
that $C'$ vanishes at $x_{2j}$ and that $D$ vanishes at
$x_0,x_1,x_{2j-1}$.  Applying a permutation matrix to
\begin{equation}
\mathcal{D}_{\beta}=
\setcounter{MaxMatrixCols}{11}W^{1/2}
\begin{pmatrix}
F_0    & 0      & C'_0   & 0      & B_0    &        &        &        &        &        &        \\
0      & F_1    & 0      & C'_1   & 0      & B_1    &        &        &        &        &        \\
C'_2   & 0      & F_2    & 0      & D_2    & 0      & A_2    &        &        &        &        \\
0      & C'_3   & 0      & F_3    & 0      & D_3    & 0      & A_3    &        &        &        \\
A_4    & 0      & D_4    & 0      & F_4    & 0      & C'_4   & 0      & B_4    &        &        \\
0      & A_5    & 0      & D_5    & 0      & F_5    & 0      & C'_5   & 0      & B_5    &        \\
        &        & \ddots &        & \ddots &        & \ddots &        & \ddots &        & \ddots
\end{pmatrix}W^{-1/2},
\end{equation}
where $F_n=-(A_n+B_n+C'_n+D_n)$, gives
\begin{equation}
\widetilde{\mathcal{D}_\beta}
=
W^{1/2}
\left(
\begin{array}{*{7}{c}:*{5}{c}}
F_0 & C'_0 & B_0 & & & & 
& \multicolumn{5}{c}{} \\

C'_2 & F_2 & D_2 & A_2 & & &
& \multicolumn{5}{c}{} \\

A_4 & D_4 & F_4 & C'_4 & B_4 & &
& \multicolumn{5}{c}{} \\

0 & B_6 & C'_6 & F_6 & D_6 & A_6 &
& \multicolumn{5}{c}{0} \\

& \ddots & \ddots & \ddots & \ddots & \ddots &
& \multicolumn{5}{c}{} \\

& & & B_{2j-2} & C'_{2j-2} & F_{2j-2} & D_{2j-2}
& \multicolumn{5}{c}{} \\

& & & 0 & A_{2j} & D_{2j} & F_{2j}
& \multicolumn{5}{c}{} \\

\hdashline

\multicolumn{7}{c:}{}
& F_1 & C'_1 & B_1 & & \\

\multicolumn{7}{c:}{}
& C'_3 & F_3 & D_3 & A_3 & \\

\multicolumn{7}{c:}{0}
& A_5 & D_5 & F_5 & C'_5 & B_5 \\

\multicolumn{7}{c:}{}
& \ddots & \ddots & \ddots & \ddots & \ddots

\end{array}
\right)
W^{-1/2}.
\end{equation}
which is block diagonal.  The parameters of the two complementary
Bannai--Ito submodules are
\begin{equation}
    \widetilde{\rho}^e_1=\frac{b-j+a-1}{4}, \qquad
    \widetilde{\rho}^e_2=\frac{b-j-a-1}{4}, \qquad
    \widetilde{r}^e_1=\frac{j+b-a+1}{4}, \qquad
    \widetilde{r}^e_2=\frac{j+b+a+1}{4},
    \label{eBIe}
\end{equation}
\begin{equation}
    \widetilde{\rho}^o_1=\frac{b-j+a-1}{4}, \qquad
    \widetilde{\rho}^o_2=\frac{b-j-a-1}{4}, \qquad
    \widetilde{r}^o_1=\frac{j+b+a+1}{4}, \qquad
    \widetilde{r}^o_2=\frac{j+b-a+1}{4},
    \label{eBIo}
\end{equation}
and the limit \eqref{limitpbi} gives
\begin{equation}
w_n^e(x_{2s})=(1-\sigma)\,\boldsymbol{P}^{(0)}_n(x_{2s})
+\sigma\,\boldsymbol{P}^{(0)}_{2j-n}(x_{2s}),
\end{equation}
\begin{equation}
w_n^o(x_{2s+1})=\sigma\,
\frac{\left(\frac{a-j}{2}\right)_{\left\lceil n/2\right\rceil}
}{\left(-\frac{a+j}{2}\right)_{\left\lceil n/2\right\rceil}
}\left[\boldsymbol{P}^{(0)}_n(x_{2s+1})-\boldsymbol{P}^{(0)}_{2j-n}(x_{2s+1})
\right],
\end{equation}
with $w_n$ the non-monic complementary Bannai--Ito polynomial.  Inverting,
\begin{equation}
\boldsymbol{P}^{(0)}_n(x_{2s})=\begin{cases}
w_n^e(x_{2s}), & n\leq j,\\[2mm]
w_{2j-n}^e(x_{2s}), & n>j,
\end{cases}
\end{equation}
\begin{equation}
\boldsymbol{P}^{(0)}_n(x_{2s+1})
=
\begin{cases}
\dfrac{j-2m}{j}
\dfrac{
\left(-\frac{a+j}{2}\right)_m
}{
\left(1+\frac{a-j}{2}\right)_m
}
w_{2m}^o(x_{2s+1}),
& n=2m\leq j,\\[5mm]
-\dfrac{2}{j}
\dfrac{
\left(-\frac{a+j}{2}\right)_{m+1}
}{
\left(1+\frac{a-j}{2}\right)_m
}
w_{2m+1}^o(x_{2s+1}),
& n=2m+1\leq j,\\[5mm]
\dfrac{1-\sigma}{\sigma}
(-1)^{2j-n}
\dfrac{j-2m}{j}
\dfrac{
\left(-\frac{a+j}{2}\right)_m
}{
\left(1+\frac{a-j}{2}\right)_m
}
w_{2m}^o(x_{2s+1}),
& 2j-n=2m,\ n>j,\\[5mm]
-\dfrac{1-\sigma}{\sigma}
(-1)^{2j-n}
\dfrac{2}{j}
\dfrac{
\left(-\frac{a+j}{2}\right)_{m+1}
}{
\left(1+\frac{a-j}{2}\right)_m
}
w_{2m+1}^o(x_{2s+1}),
& 2j-n=2m+1,\ n>j.
\end{cases}
\end{equation}
\subsubsection{Irreducibility of the submodules}
\label{irredPBI}
Checking if the two summands can be reduced further comes back to identifying the roots of the Dunkl-Shift operator coefficients. In the odd case, both coefficients are quadratic, and the roots associated to each one have already been identified to be the first two spectral points in the case of $G(x)$ and the last two spectral points in the case of $F(x)$. Hence, there are no other points in the spectrum where the difference operator coefficients give zero, assuring the irreducibility of the submodules.

In the even case, $A(x)$ and $B(x)$ are both quartic, $D(x)$ is cubic. All of their roots have already been associated to spectral points allowing the rearrangement seen above. $C(x)$ has the root $x_{2j}$. However, it also potentially has the roots
\begin{equation}
x
=
\pm \frac{1}{2}
\sqrt{
\frac{
2j^{2}\rho_{1}+j^{2}
+4j\rho_{1}^{2}
+4j\rho_{1}\rho_{2}
+6j\rho_{1}
+2j\rho_{2}
+4j
+8\rho_{1}^{2}\rho_{2}
+4\rho_{1}^{2}
+8\rho_{1}\rho_{2}
+4\rho_{1}
+4\rho_{2}
}{
2j+2\rho_{2}-1
}
}.
\label{exceptionPBI}
\end{equation}
If a spectral point $x_s, s\neq 2j$ is equal to one of these roots, $C(x)$ cancels elsewhere on the spectrum than has already been established, possibly implying the reducibility of the submodules.

\begin{theorem}
\label{thm:PBI}
Let $N=2j+p$ with $p=0,1$ and $j$ even.  The $(N+1)$-dimensional module
carried by the para-Bannai--Ito polynomials is the direct sum of two modules
of the same algebra: two Bannai--Ito modules of dimension $j+1$ with the
parameters \eqref{oBIe}, \eqref{oBIo} when $p=1$, and two complementary
Bannai--Ito modules of dimensions $j+1$ and $j$ with the parameters
\eqref{eBIe}, \eqref{eBIo} when $p=0$.  By Proposition \ref{prop:irredgen},
each summand is irreducible if and only if there is no integer $s \in \{0,1,\ldots,N\}\setminus\{2j\}$ such that $x_s$ is equal to \eqref{exceptionPBI}.  For $j$ odd the argument is the same, with the
parametrizations of \cite{paraBI}.
\end{theorem}

\section{The para-Krawtchouk polynomials within the Bannai--Ito families}
\label{BIreal}

The Bannai--Ito polynomials and their complementary family are the
eigenfunctions of Dunkl shift operators, built with a reflection, and arise
as $q\to-1$ limits of the Askey--Wilson polynomials \cite{BI,CBI}.  They have
already appeared in Subsection \ref{PBIsec}, where the para-Bannai--Ito
polynomials were obtained by truncating them.  The relation examined here is
a different one and should not be confused with it.  There, a truncation was
applied to these families and a new family came out.  Here, no truncation is
applied: the para-Krawtchouk polynomials are found to be, as they stand,
members of these families at particular values of the four parameters
$\rho_1,\rho_2,r_1,r_2$, two of which vanish \cite{parakrawtchouk}.  The two
relations also cross over in parity, which is worth noticing at the outset:
the para-Bannai--Ito polynomials come from the Bannai--Ito family when $N$ is
odd and from the complementary one when $N$ is even, whereas the
para-Krawtchouk polynomials sit inside the complementary family when $N$ is
odd and inside the Bannai--Ito family when $N$ is even.

The identification raises a difficulty.  Under a truncation condition, the
$(N+1)$-dimensional module carried by the complementary Bannai--Ito
polynomials is irreducible \cite{CBI}.  The para-Krawtchouk module is one of
these, and would therefore have to be irreducible --- against Theorem
\ref{thm:PK}.

Both statements are correct.  Irreducibility is said of a module over an
algebra, and there are two algebras at play here: the complementary
Bannai--Ito algebra, which contains an involution, and the Hahn algebra of
Section \ref{Hahn algebra}, which does not.  The module is irreducible over
the first and decomposes over the second.  Subsection \ref{involution} makes
this precise; the two subsections that follow first establish the
identification for each parity.

\subsection{Inside the complementary Bannai--Ito family, $N$ odd}

The complementary Bannai--Ito polynomials \cite{CBI} satisfy the recurrence relation \eqref{CBIrec} with the recurrence coefficients given by \eqref{CBItaueven}, \eqref{CBItauodd}. Setting, for $N=2j+1$,
\begin{equation}
r_2=\rho_2=0,
\qquad
r_1=\frac{N+1+\gamma}{4},
\qquad
\rho_1=\frac{\gamma-N-3}{4},
\label{CBIparam}
\end{equation}
the recurrence relation becomes
\begin{equation}
W_{n+1}(x)
+
\frac{n(N+1-n)\left((N+1-2n)^2-\gamma^2\right)}
{16(N-2n)(N+2-2n)}
W_{n-1}(x)
=
xW_n(x),
\end{equation}
which is that of the para-Krawtchouk polynomials up to an affine change of
variable:
\begin{equation}
P_n(x)
=
2^n\,
W_n\!\left(
\frac{x}{2}
-\frac{N-1+\gamma}{4}
\right).
\end{equation}
The resolution lies in the coefficients $\tau_n$.  For $N$ odd, the
truncation conditions $\tau_0=\tau_{N+1}=0$ are met \cite{CBI} when one of
\begin{equation}
\text{(i)}\quad r_1-\rho_1=\frac{N+2}{2},
\qquad
\text{(ii)}\quad r_1+r_2=\frac{N+1}{2},
\qquad
\text{(iii)}\quad r_2-\rho_1=\frac{N+2}{2}
\end{equation}
holds.  The parameters \eqref{CBIparam} satisfy the first of these, so that
the family terminates at degree $N$ and the para-Krawtchouk polynomials span
an $(N+1)$-dimensional complementary Bannai--Ito module.

\subsection{Inside the Bannai--Ito family, $N$ even}

The Bannai--Ito polynomials \cite{BI} satisfy the recurrence relations given by \eqref{BIrec} with coefficients \eqref{BIAn}, \eqref{BICn}. Setting, for $N=2j$,
\begin{equation}
r_2=\rho_2=0,
\qquad
r_1=\frac{N-\gamma+2}{4},
\qquad
\rho_1=\frac{-N-\gamma}{4},
\label{BIparam}
\end{equation}
the recurrence relation becomes
\begin{equation}
    B_{n+1}+\frac{1}{4}\left(\frac{(N+1)(\gamma-1)}
    {(2n-N-1)(2n+1-N)}-1\right)B_n
    +\frac{n(N+1-n)\left((2n-N-1)^2-(\gamma-1)^2\right)}
    {16(2n-N-1)^2}B_{n-1}=xB_n,
\end{equation}
so that
\begin{equation}
\widetilde{P}_n(x)
=
2^n\,
B_n\!\left(
\frac{x}{2}
-\frac{N+\gamma}{4}
\right).
\end{equation}
Here again the parameters are not generic.  For $N$ even, $A_NC_{N+1}=0$
holds when one of
\begin{equation}
\begin{aligned}
&\text{(i)}\quad r_1-\rho_1=\frac{N+1}{2},
\qquad
\text{(ii)}\quad r_2-\rho_1=\frac{N+1}{2}, \\
&\text{(iii)}\quad r_1-\rho_2=\frac{N+1}{2},
\qquad
\text{(iv)}\quad r_2-\rho_2=\frac{N+1}{2}\,,
\end{aligned}
\end{equation}
is satisfied, and \eqref{BIparam} satisfies the first.  The family terminates
at degree $N$ and the even para-Krawtchouk polynomials span an
$(N+1)$-dimensional Bannai--Ito module.

\subsection{The involution}
\label{involution}

We can now say what separates the two statements of the opening. The complementary Bannai--Ito algebra is generated by three elements
$\kappa_1$, $\kappa_2$, $\kappa_3$ together with an involution $r$,
$r^2=I$, subject to \cite{CBI}
\begin{equation}
\begin{aligned}
{}[\kappa_1,r]&=0, \qquad \{\kappa_2,r\}=2\delta_3, \qquad
\{\kappa_3,r\}=0, \qquad [\kappa_1,\kappa_2]=\kappa_3,\\[1ex]
[\kappa_1,\kappa_3]&=\tfrac{1}{2}\{\kappa_1,\kappa_2\}
-\delta_2\kappa_3r-\delta_3\kappa_1r+\delta_1\kappa_2-\delta_1\delta_3r,\\[1ex]
[\kappa_3,\kappa_2]&=\tfrac{1}{2}\kappa_2^2+\delta_2\kappa_2^2r
+2\delta_3\kappa_1r+2\delta_3\kappa_3r+\kappa_1+\delta_4r+\delta_5 .
\end{aligned}
\label{CBIalgebra}
\end{equation}
It is realized by $\kappa_1=\mathcal{D}_\alpha$, the second order Dunkl shift operator given in \cite{CBI}, $\kappa_2=x$, a third
operator $\kappa_3$, and the involution
$P=R+\frac{\rho_2}{x}(I-R)$, with structure constants \cite{CBI}
\begin{equation}
\begin{aligned}
\delta_1 &= \alpha(g-\alpha+1),
\qquad
\delta_2 = g-2\alpha+\tfrac{3}{2},
\qquad
\delta_3 = \rho_2, \\
\delta_4 &= \alpha\big(2\rho_2^2-\rho_2+\tfrac{1}{2}\big)
+\frac{\rho_2\omega}{4}
+\frac{8\rho_1r_1r_2+4r_1r_2-2\rho_1+2r_1+2r_2-3}{8}.
\end{aligned}
\label{CBIstruct}
\end{equation}
Two facts settle the matter. First, the involution crosses the sub-lattices.  On the Bannai--Ito grid,
$Rx_k=x_{k-1}$ for $k$ even and $x_{k+1}$ for $k$ odd \cite{CBI}: the
reflection exchanges the even and odd sites.  It therefore violates
condition (ii) of Proposition \ref{prop:mechanism}, while $\kappa_1$ and
$\kappa_2$ satisfy it.  A module carrying $r$ has no reason to split, and it
does not. Second, at the para-Krawtchouk values the involution drops out of the
relations.  The operator that the para-Krawtchouk polynomials diagonalize is
$\mathcal{D}_\alpha$ at $\alpha=\frac{1-N}{4}$ \cite{CBI}, and the
parameters \eqref{CBIparam} give $\rho_2=0$ and
\begin{equation}
g=\rho_1+\rho_2-r_1-r_2=-\frac{N+2}{2},
\end{equation}
whence, by \eqref{CBIstruct},
\begin{equation}
\delta_2=\delta_3=\delta_4=0.
\end{equation}
The relations \eqref{CBIalgebra} therefore close on
$\kappa_1$, $\kappa_2$ alone:
\begin{equation}
\begin{aligned}
[\kappa_1,[\kappa_1,\kappa_2]]&=\tfrac{1}{2}\{\kappa_1,\kappa_2\}+\delta_1\kappa_2,\\
[[\kappa_1,\kappa_2],\kappa_2]&=\tfrac{1}{2}\kappa_2^2+\kappa_1+\delta_5,
\end{aligned}
\end{equation}
which are the relations of the Hahn algebra \eqref{hahnalgebra}, and $r$ no
longer appears.

The situation is thus the following.  The same $(N+1)$-dimensional space
carries a module for two algebras, one contained in the other: the
complementary Bannai--Ito algebra, and the Hahn subalgebra obtained by
dropping the involution.  Over the larger algebra the module is irreducible,
because the involution moves each sub-lattice onto the other and no proper
subspace is stable.  Over the subalgebra it is the direct sum of two
irreducible Hahn modules, by Theorem \ref{thm:PK}.  Passing from an
irreducible representation to a reducible one by restricting to a subalgebra
is a standard situation.  Here the element that carries it is the involution:
it holds the two halves of the bi-lattice together, and the decomposition of
Section \ref{reducibility} is what is left when it is removed.

For $N$ even the account is not the same, and the difference is worth
setting out.  The Bannai--Ito algebra has no involution among its
generators: it is generated by $X$, $Y$, $Z$ subject to \eqref{BIalgebra}.
The reflection is present all the same, inside $X=2L+\kappa$, since the Dunkl
shift operator \eqref{BIdunkl} is built from $R$ and $T^{+}R$.  At the values
\eqref{BIparam} one finds
\begin{equation}
\kappa=\rho_1+\rho_2-r_1-r_2+\tfrac12=-\frac{N}{2},
\end{equation}
and $X$ is diagonal on the polynomials with eigenvalues
$\mu_n=(-1)^n(n+\kappa)$ \cite{BI}.  The operator $Y$ of Section
\ref{Hahn algebra} is diagonal on the same basis with eigenvalues
$\lambda_n=2n(n-N)$, and
\begin{equation}
\lambda_n=2\mu_n^2-\frac{N^2}{2},
\qquad\text{that is}\qquad
Y=2X^2-\frac{N^2}{2}\,I .
\label{YX2}
\end{equation}
The Hahn operator is therefore the square of the Bannai--Ito generator, not
the generator itself.  Squaring disposes of the reflection, $R^2=I$, and with
it of the crossing of the sub-lattices; the Hahn algebra is here the
subalgebra generated by the multiplication operator and by $X^2$, and it does
not contain $X$.

The two parities thus reach the same end by two routes.  In both, the
operator that crosses the sub-lattices is what leaves: for $N$ odd because
the two structure constants carrying the involution vanish, for $N$ even
because one passes to the square of the generator carrying the reflection.

\section{Occurrences}
\label{occurrences}

This section sets out the constructions through which the para families
present themselves, and the settings in which they have been put to use.

\subsection{Constructions}
\label{constructions}

These families present themselves along several routes.  Three of them have
already been traveled in the preceding sections and are recalled here only
so that they may be set beside the others; the remaining three are taken up
in the subsections that follow.

\paragraph{Routes already traveled.}
The first is the inverse spectral problem, by which the para-Krawtchouk
polynomials were met: a spectrum is posited and the Jacobi matrix having it
is reconstructed by the orthogonal polynomials attached to it, the perfect
transfer condition leaving the shape of the grid free \cite{parakrawtchouk}.
Nothing hypergeometric is assumed at the outset there, and the hypergeometric
representation is found afterwards.  The second is the singular
truncation of Section \ref{generalprop} and of Section \ref{GM}, which
terminates the recurrence relation without terminating the hypergeometric
series at the same stroke.  It is singular in that it introduces a pole in
the recurrence coefficients near $n=N/2$; a parametrization followed by a
limit removes the pole, and the splitting of the polynomials according to the
degree is what remains of it.  The third consists of the limits internal to the scheme,
established in Subsections \ref{PRsec} and following: the $q$-para-Racah
polynomials stand at the top, the para-Racah polynomials are their $q\to1$
limit and the para-Bannai--Ito polynomials their $q\to-1$ limit, while the
para-Krawtchouk polynomials follow from the para-Racah ones as the quadratic
bi-lattice degenerates to a linear one.

The first two deserve a remark.  They are unrelated undertakings: one starts
from a spectrum and knows nothing of hypergeometric functions, the other
starts from a hypergeometric family and cuts it.  That they arrive at the
same polynomials is a fact about these families, not a foregone conclusion.

The three routes that follow have not been used so far in this paper.

\paragraph{Removal of a pole.}
In the construction of $XX$ chains of $q$-Racah type, the second persymmetry
condition, $\alpha\beta q^{N+1}=1$, makes the recurrence coefficients
diverge.  A suitable parametrization followed by a limit removes the pole,
and what comes out is the family of $q$-para-Racah polynomials, together with
a second family of chains with perfect state transfer \cite{qRacahPST}.  Here
the para family is the regularization of a classical one at a singular point
of its parameter space.  Subsection \ref{PSTsec} returns to this from the
side of the physical model.

\paragraph{Tridiagonal representations of quadratic algebras.}
One may instead start from an algebra with two generators and ask for the
representations in which one generator acts diagonally and the other
tridiagonally.  The overlaps between the two bases are then orthogonal
polynomials, and which family they are is decided by the parameters of the
algebra.  The deformed Jordan plane yields the para-Krawtchouk polynomials at
one value of its deformation parameter and the ordinary Hahn polynomials at
another \cite{jordan}.  A related construction, in which the generators are
assembled from a five-dimensional set of elementary difference operators,
gives the para-Krawtchouk polynomials as the basis of a finite-dimensional
representation of a degenerate Sklyanin algebra \cite{sklyanin}.  Subsections
\ref{skl} and \ref{jordanplane} set both out.

\paragraph{Darboux transformations.}
Removing the last spectral point by a Christoffel transform takes the family
with $N$ odd to the family with $N$ even, as in Section \ref{generalprop},
and preserves perfect state transfer \cite{PST_construct}.  More generally,
periodic reductions of the factorization chain, which are Darboux
transformations of the Jacobi matrix taken with a period, were seen to
produce the Hahn algebra together with sub-lattices of the spectral points
\cite{SVZ}.  The para-Krawtchouk polynomials are an instance of such a
closure, with period four.

\subsection{Perfect state transfer}
\label{PSTsec}

The para-Krawtchouk polynomials were found as a solution of the perfect state
transfer conditions \cite{parakrawtchouk}.  For the XX spin chain
\begin{equation}
H =
\frac{1}{2}\sum_{l=0}^{N-1}
J_{l+1}
\left(
\sigma_l^x \sigma_{l+1}^x
+
\sigma_l^y \sigma_{l+1}^y
\right)
+
\frac{1}{2}\sum_{l=0}^{N}
B_l
\left(
\sigma_l^z + 1
\right),
\end{equation}
the one-excitation dynamics is carried by the Jacobi matrix $J$ with entries
$J_l$ and $B_l$, and perfect state transfer means that
\begin{equation}
    e^{iTJ}\ket{e_0}=e^{i\phi}\ket{e_N}
\end{equation}
for some time $T$ and some real phase $\phi$.  This holds if and only if
\cite{kay}
\begin{enumerate}[label=(\roman*)]
    \item the eigenvalues satisfy $x_{s+1}-x_s=\frac{\pi}{T}M_s$ with $M_s$
    positive odd integers, and
    \item $J$ is mirror-symmetric, $RJR=J$, with $R$ as in
    \eqref{reflexmatrix}.
\end{enumerate}
The second condition is equivalent to either of
\begin{equation}
\text{(ii$'$)}\quad
w_s=\frac{1}{\left|P_{N+1}'(x_s)\right|}>0,
\qquad\qquad
\text{(ii$''$)}\quad
P_N(x_s)=(-1)^{N+s}\,,
\end{equation}
as shown in \cite{PST_construct}.
These two forms are
constructive: from a spectrum satisfying (i), the polynomials $P_{N+1}$ and
$P_N$ are written down using (ii$'$), and the remaining $P_n$ and the matrix
$J$ follow by Euclidean division.  Applied to the bi-lattice
\eqref{evenspec}--\eqref{oddspec}, this construction returns the
para-Krawtchouk polynomials, the transfer condition being met when
$\gamma=M_1/M_2$ with $M_1$, $M_2$ coprime positive integers and $M_1$ odd.
The same circle of ideas yields chains with almost perfect state transfer,
where the transfer probability can be brought arbitrarily close to one
without the spectral condition (i) being met exactly \cite{APST}.

The other para families enter the same problem on their own grids.  Taking
the couplings from the recurrence coefficients of the $q$-Racah polynomials
gives a chain with perfect state transfer whose one-excitation spectrum is
$q$-quadratic; the persymmetry conditions that the couplings must satisfy
come in two kinds, and the second of them makes the coefficients diverge.
The $q$-para-Racah polynomials remove that divergence and furnish a second
family of such chains, with eigenvalues and parameter constraints of their
own \cite{qRacahPST}.  This is the construction described in Subsection
\ref{constructions}, seen from the side of the physical model.

\subsection{Fractional revival}

Fractional revival is the weaker requirement that the excitation be
distributed coherently between the two ends of the chain at some time, and
reconstituted later.  Couplings realizing it had been obtained numerically
\cite{BCB}.  The recurrence coefficients of the para-Krawtchouk polynomials
provide an analytic model of the effect, with two parameters and a revival
time that does not depend on the length of the chain \cite{revival}.  The
protocol of \cite{XKT} combines fractional revival on such a chain with
dual-rail encoding to transfer a state perfectly in an average time below the
bound that holds for a single chain. Analogous
analytic models have since been obtained from the para-Racah
\cite{LemayVinetZhedanov2016} and $q$-para-Racah \cite{SchererVinetZhedanov2022}
polynomials.

\subsection{Classical mass-spring chains}

The two effects above are not intrinsically quantum.  A chain of masses
joined by springs, with masses $m_l$ and spring constants $k_l$ suitably
chosen, transmits a pulse from one end to the other without dispersion, and
redistributes momentum periodically between the two ends: a Newton's cradle
whose behaviour is exact rather than approximate.  The equations of motion of
such a chain reduce, in the harmonic regime, to an eigenvalue problem for a
Jacobi matrix, and the same inverse spectral problem is at work.  The
para-Racah polynomials determine the masses and spring constants of chains
with perfect transfer and fractional revival \cite{cradle}, and the $q$-Racah
polynomials do the same on a $q$-quadratic grid \cite{cradleq}.  The para
families thus describe classical mechanical devices as readily as quantum
spin chains, the polynomials being attached to the Jacobi matrix rather than
to any particular physical interpretation of it.

\subsection{The degenerate Sklyanin algebra}
\label{skl}

This is the setting in which the para-Krawtchouk polynomials were given the
algebraic interpretation they had been lacking \cite{sklyanin}.  We recall
the construction, and how the para-Krawtchouk polynomials enter~it.

An S-Heun operator on the linear grid \eqref{SHeundef}
is a difference operator $S=A_1T_++A_2T_-$, without diagonal term, that sends
polynomials of degree $n$ to polynomials of degree at most $n+1$.  Imposing
that condition on $1$ and on $x$ fixes $A_1$ and $A_2$ up to five free
parameters, so that the S-Heun operators form a five-dimensional linear
space.  A basis is
\begin{subequations}\label{eq:operators}
\begin{align}
L &= \tfrac{1}{2}\left[T_{+}-T_{-}\right], \\
M_{1} &= \tfrac{1}{2}\left[T_{+}+T_{-}\right], \\
M_{2} &= \tfrac{1}{2}x\left[T_{+}-T_{-}\right], \\
R_{1} &= \tfrac{1}{2}x\left[(1-2x)T_{+}+(1+2x)T_{-}\right], \\
R_{2} &= \tfrac{1}{2}x\left[T_{+}+T_{-}\right],
\end{align}
\end{subequations}
with $T_{\pm}f(x)=f(x\pm1)$.  The operator $L$ lowers the degree of a
polynomial by one, $M_1$ and $M_2$ leave it unchanged, and $R_1$, $R_2$ raise
it by one.

Take now the three operators that leave the degree unchanged.  Using the
quadratic relations among the generators, the most general quadratic
expression in $L$, $M_1$, $M_2$ can be brought to
\begin{equation}
Q = \alpha_1 L^2 + \alpha_2 LM_1 + \alpha_3 LM_2 + \alpha_4 M_1^2
  + \alpha_5 M_1M_2 + \alpha_6 M_2^2 .
\label{Qstab}
\end{equation}
Substituting the expressions \eqref{eq:operators} makes $Q$ a second-order
difference operator, and its eigenvalue equation is, after rearrangement,
the difference equation \eqref{diffc_hahn} of the continuous Hahn
polynomials, the four parameters $a,b,c,d$ being determined by the
$\alpha_i$.  This is what is meant by saying that the continuous Hahn
polynomials come from quadratic combinations of S-Heun operators: they are
the eigenfunctions of the most general quadratic expression in the three
degree-preserving generators.

The other bispectral operator is obtained just as concretely.  Multiplication
by the variable is the commutator of two of the generators,
\begin{equation}
X=[M_{2},R_{2}].
\label{XSHeun}
\end{equation}
Both bispectral operators are therefore quadratic in the S-Heun basis.  So
are the four structure operators of Kalnins and Miller --- the forward,
backward and two contiguity operators --- which are linear combinations of
the five generators; and the Heun operator attached to the continuous Hahn
polynomials is recovered as the most general quadratic combination that
raises the degree by at most one.  The five S-Heun operators are thus the
elementary blocks from which the whole apparatus of the family is assembled.

We come to the algebra.  Writing $\nu=\frac{1}{2}(N-1)$ and setting
\begin{equation}
\begin{aligned}
A &= 2(\nu+1)M_{1}-2M_{2},
&\qquad
B &= \tfrac{1}{2}(2\nu+1)(2\nu+3)L-R_{1}-(4\nu+3)R_{2},\\
C &= L,
&\qquad
D &= M_{1},
\end{aligned}
\end{equation}
one obtains the quadratic relations
\begin{equation}
\begin{aligned}
[C,D] &= 0,
& [A,C] &= \{C,D\},
& [A,D] &= \{C,C\},\\
[B,C] &= \{D,A\},
& [B,D] &= \{C,A\},
& [B,A] &= \{B,D\},
\end{aligned}
\end{equation}
which define $Skl_4$, a degeneration of the Sklyanin algebra.

The realization above involves $\nu$, and it yields a finite-dimensional
representation precisely when $\nu$ is an integer or a half-integer, that is,
when
\begin{equation}
1-(a+b+c+d)=N,
\qquad N\in\mathbb{Z}_{>0}.
\end{equation}
This is the truncation condition \eqref{truncation} of Section
\ref{generalprop}.  What was introduced there as a way of terminating the
recurrence relation of the continuous Hahn polynomials is, on the algebraic
side, the condition for $Skl_4$ to have a representation of dimension $N+1$;
and the basis of that representation is the para-Krawtchouk family.  This is
the algebraic interpretation of these polynomials, and it had been missing
until \cite{sklyanin}.

The two bispectral operators can then be written out in the S-Heun basis.
Using \eqref{parametrization} and \eqref{gamma}, the parameters are related
by
$a+d=-j-\gamma/2$, and writing $a-d=\rho$, the continuous Hahn difference
operator reads

\begin{equation}
\mathcal D={}
\left[(j+1)^2+\rho^2-\frac{\gamma^2}{4}\right]L^2
+2\rho\left[(j+1)LM_1-LM_2\right]
-\frac12 LR_1
-\left(2j+\frac32\right)LR_2
+2(j+1).
\end{equation}

We are able to recover the PK difference operator once applying the translation $y=x+j-\rho+\frac{\gamma}{2}$.
In that case, we obtain that 
\begin{equation}
    Y=2\mathcal{D}.
\end{equation}
for $N=2j+1$, while multiplication by $x$ is \eqref{XSHeun}.  The same
construction on the $q$-linear grid gives the $q$-para-Krawtchouk polynomials
as the basis of a finite-dimensional representation of $U_q(sl_2)$
\cite{sklyanin}.
\subsection{The deformed Jordan plane}
\label{jordanplane}

The para-Krawtchouk polynomials also appear among the tridiagonal
representations of quadratic algebras \cite{jordan}.  Let $\mathcal{A}$ be
generated over $\mathbb{R}$ by $X$ and $Z$ subject to
\begin{equation}
[Z,X]=Z^2+\Delta,
\end{equation}
and let the generators act on an orthonormal basis $\ket{n}$ by
\begin{equation}
\begin{aligned}
X\ket{n}
&= c_n\ket{n-1}+b_n\ket{n}+a_n\ket{n+1},\\
Z\ket{n}
&= u_n\ket{n-1}+v_n\ket{n}+w_n\ket{n+1}.
\end{aligned}
\end{equation}
The overlaps $q_n(x)=\braket{x|n}$ with the eigenvectors of $X$ are
polynomials, and their monic versions $p_n$ satisfy a recurrence relation
whose coefficients are determined by the algebra:
\begin{equation}
a_{n-1}c_n
=
(n+\phi_0)(n-\delta_0-1)
\frac{
n(n+\phi_0-\delta_0-1)
\left[
\Delta(2n+\phi_0-\delta_0-1)^2
+
v_0^2(\phi_0-\delta_0+1)^2
\right]
}{
(2n+\phi_0-\delta_0-1)^2
(2n+\phi_0-\delta_0)
(2n+\phi_0-\delta_0-2)
},
\end{equation}
\begin{equation}
b_n
=
\frac{1}{2}
\frac{
(\delta_0+\phi_0+1)
(\phi_0-\delta_0+1)
(\phi_0-\delta_0-1)v_0
}{
(2n+\phi_0-\delta_0-1)
(2n+\phi_0-\delta_0+1)
}
+\widetilde{b}_0,
\end{equation}
with $\phi_n=c_n/u_n=\phi_0+n$, $\delta_n=a_n/w_n=\delta_0-n$ and
$\widetilde{b}_0=b_0-\frac{1}{2}(\delta_0+\phi_0+1)v_0$.  Take $\Delta=-1$,
$N=2j+p$ and
\begin{equation}
\phi_0=-j-1+\frac{\gamma}{2}+t,
\quad
-\delta_0=-j+1-\frac{\gamma}{2}-p+t,
\quad
v_0=\frac{1-p}{N-1-2t},
\quad
\widetilde{b}_0=\frac{1}{2}(N+\gamma-1).
\label{pkquad}
\end{equation}
In the limit $t\to0$ these give the para-Krawtchouk recurrence coefficients
\eqref{recurrenceB}--\eqref{recurrenceU} for either parity.

The same construction with $\Delta=-\frac{1}{4}$ and
\begin{equation}
\phi_0=\beta,
\qquad
-\delta_0=\alpha+1,
\qquad
v_0=-\frac{\alpha+\beta+2N+2}{2(\alpha+\beta+2)},
\qquad
\widetilde{b}_0=\frac{1}{4}(2N-\alpha+\beta)
\label{hahnquad}
\end{equation}
gives the Hahn polynomials $Q_n(x;\alpha,\beta,N)$.  Substituting for
$\alpha$, $\beta$ and $N$ the parameters of the submodules,
\eqref{oparame} and \eqref{oparamo}, turns \eqref{hahnquad} into
\eqref{pkquad}: the decomposition of Section \ref{reducibility} is visible in
this setting as well.

\subsection{Orthogonal polynomials on bi-lattices}
\label{bilattices}

Orthogonal polynomials on grids formed of two interlaced lattices were
studied before the para families, in connection with the Meixner and Charlier
polynomials and with birth and death processes \cite{SVA}.  What the para
families add to that setting is bispectrality: the difference equation gives
the operator $Y$ of Section \ref{Hahn algebra} and, through it, the
decomposition established here, and a family orthogonal on a bi-lattice but
not bispectral has no such operator.

The closer neighbour is the doubling of Oste and Van der Jeugt
\cite{OVdJdouble}.  Two families of one classical type whose parameters
differ by units are made to obey a common pair of relations,
\begin{equation}
a(n)y_n+b(n)y_{n+1}=\widehat{d}(x)\,\widehat{y}_n,
\qquad
\widehat{a}(n)\widehat{y}_n+\widehat{b}(n)\widehat{y}_{n+1}=d(x)\,y_{n+1},
\label{doubling}
\end{equation}
with $a,b,\widehat a,\widehat b$ depending on $n$ alone and $d,\widehat d$ on
$x$ alone; they then assemble into a single Jacobi matrix with zero diagonal
and explicit spectrum, symmetric about the origin, from which they are
recovered as the even and odd degree parts.  The doubles of the Hahn, dual
Hahn and Racah families are classified in \cite{OVdJdouble}.  That the two
members of a double are related by a Christoffel transform holds for every
quadratic decomposition: if
\begin{equation}
Q_{2n}(x)=P_n\bigl(\pi(x)\bigr),
\qquad
\pi(x)=(x-a)(x-b)+c,
\label{quaddec}
\end{equation}
then orthogonality forces $Q_{2n+1}(x)=(x-b)P^{*}_n(c;\pi(x))$, with
$P^{*}_n(c;\,\cdot\,)$ the kernel polynomials of $P_n$, that is, its
Christoffel transform at $c$ \cite{chihara}.  What the classification determines
is when that transform stays within the family, and in every case it finds a
unit shift of the parameters.  The generalized Hermite polynomials, two
Laguerre families of parameters $\alpha$ and $\alpha-1$ under one weight, are
the classical instance \cite{rosenblum}; each double underlies a finite
oscillator model whose algebra extends $su(2)$ by an involution,
$[J_+,J_-]=2J_0+2cP$, a Bannai--Ito algebra with two parameters set to zero
\cite{OVdJsu2P}.

Both constructions carry two families of one classical type on a single
Jacobi matrix, and they differ in the direction of the split.  A double splits
by the parity of the degree --- this is \eqref{quaddec} with $a=b=c=0$,
Chihara's construction for symmetric polynomials \cite{chihara},
$p_{2n}(x)=P_n(x^2)$ and $p_{2n+1}(x)=xR_n(x^2)$ --- whereas the
decomposition of Section \ref{reducibility} splits by the parity of the
spectral point.  The para-Krawtchouk polynomials with $N$ odd admit both,
since $B_n$ is constant by \eqref{recurrenceB} and the shift
$x\to x-\frac12(N-1+\gamma)$ makes them symmetric, and following the first
splitting shows where the two constructions part.  The shifted spectrum is
$\pm\left(2s-j-\frac{\gamma}{2}\right)$, $s=0,1,\ldots,j$, so the two
families of Chihara's construction are orthogonal on the squares of
$\{|2s-j-\frac{\gamma}{2}|\}$; folding a progression of step two about the
origin returns a bi-lattice, whose gaps here alternate between $\gamma$ and
$2-\gamma$, and it is a single progression only at $\gamma=1$.  The grid of a
Hahn, dual Hahn or Racah family is a polynomial of degree at most two in the
index, with constant second differences; those of the squares of the folded
set, read off the two gaps, are constant only at $\gamma=1$.  The
para-Krawtchouk polynomials therefore constitute a double only at $\gamma=1$,
where the bi-lattice \eqref{evenspec}--\eqref{oddspec} is uniform and the
family is the symmetric Krawtchouk one --- the value at which the chain of
Subsection \ref{PSTsec} carries perfect state transfer on a uniform spectrum.

The two members of a double are separated by a unit shift of the parameters,
the shift a Christoffel transform produces when it stays within the family;
the two submodules of a para family are separated by the exchange of a pair
of parameters --- $\alpha\leftrightarrow\beta$ in
\eqref{oparame}--\eqref{oparamo}, $\delta\to-\delta$ in
\eqref{racahcoeffe}--\eqref{racahcoeffo}, $\delta\to\delta^{-1}$ in
\eqref{qracahcoeffe}--\eqref{qracahcoeffo}, $r_1\leftrightarrow r_2$ in
\eqref{oBIe}--\eqref{oBIo} --- an offset that varies continuously with the
bi-lattice.  The involution sits differently as well: it is a generator of
the doubling algebra, while the bispectral algebras of Section \ref{GM}
contain none, and Subsection \ref{involution} shows what adjoining one costs.

\section{Concluding remarks}
\label{conclusion}

The para families were built one at a time, each from its own truncation and
each on its own grid.  What has been shown is that they are not new
families in the sense in which the families of the Askey scheme are new: each
of them carries two copies of a classical family, glued along a bi-lattice.
The gluing is governed by the three features of Proposition
\ref{prop:mechanism}, and the pairing $\lambda_n=\lambda_{N-n}$ of the
eigenvalues, which had been recorded for each family separately, reflects the fact that the two summands have the same spectrum. Several questions are left open.

The three features of Proposition \ref{prop:mechanism} are sufficient for the
reduction, and every para family known to us satisfies them.  Whether they
characterize the situation, that is, whether a finite bispectral family whose
module splits in two must be of this form, remains an open question. A related instance concerns the spectrum of the
parabose oscillator, which named the para-Krawtchouk polynomials, has a
natural analogue for the parafermionic and parabosonic oscillators of higher
order. Whether their spectra admit an analogous para
construction, two copies of a classical family glued along a bi-lattice,
is a question we have not investigated.

The doubles classified in \cite{OVdJdouble} join two families of one classical
type separated by a unit shift of the parameters, and Subsection
\ref{bilattices} places the para families outside that classification except
at $\gamma=1$.  Whether the classification admits a separation that varies
continuously, and what such a construction would give at the Racah and
$q$-Racah levels, remain to be examined.  A related question concerns the two families that Chihara's construction
attaches to the para-Krawtchouk polynomials: they are orthogonal on the
squares of a bi-lattice, the shape of the para-Racah grid, and they are
related, by \eqref{quaddec}, through a Christoffel transform.  Whether they
are in fact para-Racah families is a question of the same kind.

The decomposition has been established at the level of modules.  Whether it
can be seen at the level of the algebras themselves, through a map that does
not refer to a particular representation, is a further question. A related gap concerns the para-Racah and the para-Bannai-Ito polynomials, the irreducibility of the two summands fails
on the exceptional sets \eqref{PRexcept} and \eqref{exceptionPBI}.  The decomposition of the module has
not been examined. Finally, the whole discussion has been carried out for polynomials in one
variable.  Multivariate extensions of the families of the Askey scheme are
available, and whether a para construction and its decomposition exist in
that setting is open.

\section*{Acknowledgements}

The authors are grateful to Marlon Josue Rivera Valladares for fruitful discussions. NC and RN thank the Centre de Recherches Math\'ematiques (CRM) for its hospitality. RN is supported in part by the National Science Foundation under grants PHY 2310594 and PHY 2609869. LV is funded in part through a Discovery Grant of the Natural Sciences and Engineering Research Council (NSERC) of Canada; SB, QL and LM enjoy scholarships and fellowships provided by this fund. 

\section*{Conflict of interest}

The authors declare that they have no conflict of interest.

\section*{Data availability statement}

No new data were created or analysed in this study.

\appendix

\section{The classical families}
\label{app:classical}

We collect here the classical families of polynomials that appear in Sections
\ref{reducibility} and \ref{GM}, following \cite{koekoek}.  

\subsection{Askey--Wilson polynomials}
\label{app:askeywilson}

The Askey--Wilson polynomials  are sitting atop the continuous part of the $q$-Askey scheme. They admit a simple expression in terms of the $_4\phi_3$ $q$-hypergeometric function
\begin{equation}
W_n(x;a,b,c,d\,|\,q)
=
{}_4\phi_3\!\left(
\begin{matrix}
q^{-n},\; abcdq^{n-1},\; ae^{i\theta},\; ae^{-i\theta}\\
ab,\; ac,\; ad
\end{matrix}
;q,q
\right).
\label{awdef}
\end{equation}
in the variable $x=\cos\theta$. They satisfy the three-term recurrence relation
\begin{equation}
2x\,W_n(x)
=
A_n\,W_{n+1}(x)
+
\bigl(a+a^{-1}-A_n-C_n\bigr)W_n(x)
+
C_n\,W_{n-1}(x),
\label{awrec}
\end{equation}
with
\begin{align}
A_n
&=
\frac{
(1-abq^n)(1-acq^n)(1-adq^n)(1-abcdq^{n-1})
}{
a\bigl(1-abcdq^{2n-1}\bigr)\bigl(1-abcdq^{2n}\bigr)
},
\label{awA}
\\
C_n
&=
\frac{
a(1-q^n)(1-bcq^{n-1})(1-bdq^{n-1})(1-cdq^{n-1})
}{
\bigl(1-abcdq^{2n-2}\bigr)\bigl(1-abcdq^{2n-1}\bigr)
}.
\label{awC}
\end{align}
The Askey--Wilson polynomials obey the $q$-difference equation
\begin{equation}
q^{-n}(1-q^n)(1-abcdq^{n-1})\,W_n(x)
=
A(\theta)\bigl[T_+W_n(x)-W_n(x)\bigr]
+
\overline{A}(\theta)\bigl[T_-W_n(x)-W_n(x)\bigr],
\label{awdiff}
\end{equation}
where $T_\pm$ act on $e^{\pm i\theta}$ by $T_\pm e^{i\theta}=q^\pm e^{i\theta}$,
$T_\pm e^{-i\theta}=q^{\mp}e^{-i\theta}$, and
\begin{equation}
A(\theta)
=
\frac{
(1-ae^{i\theta})(1-be^{i\theta})(1-ce^{i\theta})(1-de^{i\theta})
}{
(1-e^{2i\theta})(1-qe^{2i\theta})
},
\end{equation}
with $\overline{A}(\theta)=A(-\theta)$ its complex conjugate.

\subsection{Wilson polynomials}
\label{app:wilson}

The Wilson polynomials $\widetilde W_n(x^2;a,b,c,d)$ are defined in terms of the $_4F_3$ hypergeometric function
\begin{equation}
\widetilde W_n(x^2;a,b,c,d)=
{}_4F_3\!\left(
\begin{matrix}
-n,\; n+a+b+c+d-1,\; a-ix,\; a+ix\\
a+b,\; a+c,\; a+d
\end{matrix}
\,;\,1
\right).
\label{wilsondef}
\end{equation}
They satisfy the three-term recurrence relation
\begin{equation}
-(a^2+x^2)\,\widetilde W_n(x^2)
=
A_n\,\widetilde W_{n+1}(x^2)
-(A_n+C_n)\,\widetilde W_n(x^2)
+C_n\,\widetilde W_{n-1}(x^2),
\label{wilsonrec}
\end{equation}
with
\begin{align}
&A_n=\frac{(n+a+b+c+d-1)(n+a+b)(n+a+c)(n+a+d)}{(2n+a+b+c+d-1)(2n+a+b+c+d)},\label{wilsonA}\\
&C_n=\frac{n(n+b+c-1)(n+b+d-1)(n+c+d-1)}{(2n+a+b+c+d-2)(2n+a+b+c+d-1)}.
\label{wilsonC}
\end{align}
The Wilson polynomials also the  difference equation
\begin{equation}
n(n+a+b+c+d-1)\,\widetilde W_n(x^2)
=
\overline{\mathcal D(x)}\,\widetilde W_n\!\left((x+i)^2\right)
-\bigl(\overline{\mathcal D(x)}+\mathcal D(x)\bigr)\widetilde W_n(x^2)
+\mathcal D(x)\,\widetilde W_n\!\left((x-i)^2\right),
\label{wilsondiff}
\end{equation}
with
\begin{equation}
\mathcal D(x)=\frac{(a+ix)(b+ix)(c+ix)(d+ix)}{(2ix)(2ix+1)}.
\label{wilsonD}
\end{equation}
and $\overline{\mathcal{D}(x)}=\mathcal{D}(-x)$ its complex conjugate. 

\subsection{Continuous Hahn polynomials}
\label{app:chahn}

The continuous Hahn polynomials are given in terms of the ${}_3F_2$
hypergeometric function 
\begin{equation}
p_n(x;a,b,c,d)
=
{}_3F_2\!\left(
\begin{matrix}
-n,\; n+a+b+c+d-1,\; a+ix\\
a+c,\; a+d
\end{matrix}
\,;1
\right), \qquad n\in \mathbb{Z}_{\geq 0}.
\label{hyperc_hahn}
\end{equation}
They satisfy the three-term recurrence relation
\begin{equation}
\label{eq:3t-contHahn}
(a+ix)\,p_n(x)
=
A_np_{n+1}(x)
-
(A_n+C_n)\,p_n(x)
+
C_np_{n-1}(x),
\end{equation}
with
\begin{equation}
\begin{aligned}
A_n
&=-
\frac{(n+a+b+c+d-1)(n+a+c)(n+a+d)}
{(2n+a+b+c+d-1)(2n+a+b+c+d)},
\\[1em]
C_n
&=
\frac{n(n+b+c-1)(n+b+d-1)}
{(2n+a+b+c+d-2)(2n+a+b+c+d-1)},
\end{aligned}
\label{AC}
\end{equation}
and the second-order difference equation
\begin{equation}
n(n+a+b+c+d-1)\,p_n(x)
=
B(x)\,p_n(x+i)
-
\bigl[B(x)+D(x)\bigr]\,p_n(x)
+
D(x)\,p_n(x-i),
\label{diffc_hahn}
\end{equation}
where
\begin{equation}
\begin{aligned}
B(x) &= (c-ix)(d-ix),\\
D(x) &= (a+ix)(b+ix).
\end{aligned}
\end{equation}


\subsection{Racah}
\label{app:racah}
The Racah polynomials are given in terms of the ${}_4F_3$
hypergeometric function,
\begin{equation}
H_n(\lambda(x);\alpha,\beta,\gamma,\delta)=
{}_4F_3\!\left(
\begin{matrix}
-n,\; n+\alpha+\beta+1,\; -x,\; x+\gamma+\delta+1\\
\alpha+1,\; \beta+\delta+1,\; \gamma+1
\end{matrix}
\,;1
\right),
\qquad n=0,1,\ldots,K,
\label{racahhyp}
\end{equation}
where one of $\alpha+1$, $\beta+\delta+1$, $\gamma+1$ equals $-K$ and  $\lambda(x)=x(x+\gamma+\delta+1)$.  The
difference equation is
\begin{equation}
n(n+\alpha+\beta+1)y(x)
=
B(x)y(x+1)
-
\bigl[B(x)+D(x)\bigr]y(x)
+
D(x)y(x-1),
\label{racahdiff}
\end{equation}
where $y(x)=H_n(\lambda(x);\alpha,\beta,\gamma,\delta)$ and
\begin{equation}
\begin{aligned}
B(x)
&=
\frac{
(x+\alpha+1)(x+\beta+\delta+1)(x+\gamma+1)(x+\gamma+\delta+1)
}{
(2x+\gamma+\delta+1)(2x+\gamma+\delta+2)
},
\\[2ex]
D(x)
&=
\frac{
x(x-\alpha+\gamma+\delta)(x-\beta+\gamma)(x+\delta)
}{
(2x+\gamma+\delta)(2x+\gamma+\delta+1)
}.
\end{aligned}
\label{racahBD}
\end{equation}
The recurrence
relation reads
\begin{equation}
\lambda(x)H_n(\lambda(x))
=A_nH_{n+1}(\lambda(x))-(A_n+C_n)H_n(\lambda(x))+C_nH_{n-1}(\lambda(x)),
\end{equation}
\begin{equation}
A_n=\frac{(n+\alpha+1)(n+\alpha+\beta+1)(n+\beta+\delta+1)(n+\gamma+1)}
{(2n+\alpha+\beta+1)(2n+\alpha+\beta+2)},
\end{equation}
\begin{equation}
C_n=\frac{n(n+\alpha+\beta-\gamma)(n+\alpha-\delta)(n+\beta)}
{(2n+\alpha+\beta)(2n+\alpha+\beta+1)}.
\end{equation}

\subsection{Hahn}
\label{app:hahn}
The Hahn polynomials are given in terms of the ${}_3F_2$
hypergeometric function 
\begin{equation}
\label{eq:Hahn-pol-dfn}
Q_n(x;\alpha,\beta,K)=
{}_3F_2\!\left(
\begin{matrix}
-n,\; n+\alpha+\beta+1,\; -x\\
\alpha+1,\; -K
\end{matrix}
\,;1
\right),
\qquad n=0,1,\ldots,K.
\end{equation}
The difference equation is
\begin{equation}
n(n+\alpha+\beta+1)y(x)
=
B(x)y(x+1)-\bigl[B(x)+D(x)\bigr]y(x)+D(x)y(x-1),
\label{hahndiffeq}
\end{equation}
\begin{equation}
B(x)=(x+\alpha+1)(x-K),
\qquad
D(x)=x(x-\beta-K-1),
\label{hahndiffcoeff}
\end{equation}
where $y(x)=Q_n(x;\alpha,\beta,K)$,
and the recurrence relation reads
\begin{equation}
\label{eq:3t-hahn}
-x\,Q_n(x)=A_nQ_{n+1}(x)-(A_n+C_n)Q_n(x)+C_nQ_{n-1}(x),
\end{equation}
\begin{equation}
A_n=\frac{(n+\alpha+\beta+1)(n+\alpha+1)(K-n)}
{(2n+\alpha+\beta+1)(2n+\alpha+\beta+2)},
\qquad
C_n=\frac{n(n+\alpha+\beta+K+1)(n+\beta)}
{(2n+\alpha+\beta)(2n+\alpha+\beta+1)}.
\end{equation}

\subsection{$q$-Racah}
\label{app:qracah}
The $q$-Racah polynomials are given in terms of a $_4\phi_3$ $q$-hypergeometric function,
\begin{equation}
\label{eq:qracah-defi}
\Xi_n(\mu(x);\alpha,\beta,\gamma,\delta\mid q)=
{}_4\phi_3\!\left(
\begin{matrix}
q^{-n},\; \alpha\beta q^{n+1},\; q^{-x},\; \gamma\delta q^{x+1}\\
\alpha q,\; \beta\delta q,\; \gamma q
\end{matrix}
\,;q,q
\right),
\qquad n=0,1,\ldots,K,
\end{equation}
where one of $\alpha q$, $\beta\delta q$, $\gamma q$ equals $q^{-K}$ and $\mu(x)=q^{-x}+\gamma\delta q^{x+1}$.  The
$q$-difference equation is
\begin{equation}
q^{-n}(1-q^{n})(1-\alpha\beta q^{n+1})\,y(x)
=
B(x)y(x+1)-\bigl[B(x)+D(x)\bigr]y(x)+D(x)y(x-1),
\label{qracahdiff}
\end{equation}
with $y(x)=\Xi_n(\mu(x);\alpha,\beta,\gamma,\delta\mid q)$ and
\begin{equation}
\left\{
\begin{aligned}
B(x)&=\frac{(1-\alpha q^{x+1})(1-\beta\delta q^{x+1})
(1-\gamma q^{x+1})(1-\gamma\delta q^{x+1})}
{(1-\gamma\delta q^{2x+1})(1-\gamma\delta q^{2x+2})},
\\[2ex]
D(x)&=\frac{q(1-q^{x})(1-\delta q^{x})
(\beta-\gamma q^{x})(\alpha-\gamma\delta q^{x})}
{(1-\gamma\delta q^{2x})(1-\gamma\delta q^{2x+1})}.
\end{aligned}
\right.
\label{qracahBD}
\end{equation}
It is with \eqref{qracahdiff}--\eqref{qracahBD} that the coefficients of
Subsection \ref{qPRsec} are matched.  The recurrence relation reads
\begin{equation}
-\left(1-q^{-x}\right)\left(1-\gamma\delta q^{x+1}\right)\Xi_n(\mu(x))
=A_n\Xi_{n+1}(\mu(x))-(A_n+C_n)\Xi_n(\mu(x))+C_n\Xi_{n-1}(\mu(x)),
\end{equation}
\begin{equation}
A_n=\frac{(1-\alpha q^{n+1})(1-\alpha\beta q^{n+1})
(1-\beta\delta q^{n+1})(1-\gamma q^{n+1})}
{(1-\alpha\beta q^{2n+1})(1-\alpha\beta q^{2n+2})},
\end{equation}
\begin{equation}
C_n=\frac{q(1-q^{n})(1-\beta q^{n})
(\gamma-\alpha\beta q^{n})(\delta-\alpha q^{n})}
{(1-\alpha\beta q^{2n})(1-\alpha\beta q^{2n+1})}.
\end{equation}

\subsection{Bannai--Ito and complementary Bannai--Ito}
\label{app:BI}
The monic Bannai--Ito polynomials $B_n(x;\rho_1,\rho_2,r_1,r_2)$ \cite{BI} satisfy
\begin{equation}
B_{n+1}(x)
+
(\rho_1 - A_n - C_n)B_n(x)
+
A_{n-1}C_n B_{n-1}(x)
=
xB_n(x),
\label{BIrec}
\end{equation}
with
\begin{equation}
A_n =
\begin{cases}
\dfrac{(n+2\rho_1-2r_1+1)(n+2\rho_1-2r_2+1)}
{4(n+g+1)},
& n \text{ even}, \\[1.2em]
\dfrac{(n+2g+1)(n+2\rho_1+2\rho_2+1)}
{4(n+g+1)},
& n \text{ odd},
\end{cases}
\label{BIAn}
\end{equation}
\begin{equation}
C_n =
\begin{cases}
-\dfrac{n(n-2r_1-2r_2)}
{4(n+g)},
& n \text{ even}, \\[1.2em]
-\dfrac{(n+2\rho_2-2r_2)(n+2\rho_2-2r_1)}
{4(n+g)},
& n \text{ odd},
\end{cases}
\label{BICn}
\end{equation}
and $g$ given by \eqref{PBItruncation}. It is seen from the above formulas that the positivity condition
$u_n=A_{n-1}C_n>0$ cannot be satisfied for all $n\in\mathbb{N}$. Hence it follows that the Bannai--Ito polynomials can only
form a finite set of positive-definite orthogonal polynomials
$B_0(x),\ldots,B_N(x)$, which occurs when the "local" positivity
condition $u_i>0$ for $i\in\{1,\ldots,N\}$ and the truncation conditions
$u_0=0$, $u_{N+1}=0$ are satisfied.

If $N$ is even, it follows from \eqref{BIAn}-\eqref{BICn} that the condition
$u_{N+1}=0$ is tantamount to one of the following requirements:
\begin{equation}
\begin{aligned}
1)\quad &r_1-\rho_1=\frac{N+1}{2},
&\qquad
2)\quad &r_2-\rho_1=\frac{N+1}{2},\\[6pt]
3)\quad &r_1-\rho_2=\frac{N+1}{2},
&\qquad
4)\quad &r_2-\rho_2=\frac{N+1}{2}.
\end{aligned}
\end{equation} 
\noindent If $N$ is odd, it follows from \eqref{BIAn}-\eqref{BICn} that the condition
$u_{N+1}=0$ is equivalent to one of the following restrictions:
\begin{equation}
\begin{aligned}
\text{i)}\quad &\rho_1+\rho_2=-\frac{N+1}{2},\\[6pt]
\text{ii)}\quad &r_1+r_2=\frac{N+1}{2},\\[6pt]
\text{iii)}\quad &\rho_1+\rho_2-r_1-r_2=-\frac{N+1}{2}.
\end{aligned}
\end{equation}
If these conditions are met, this leads to the standard truncated BI polynomials. They are the eigenfunctions of the Dunkl shift operator
\cite{BI}
\begin{equation}
L=F(x)(I-R)+G(x)\left(T^{+}R-I\right),
\qquad
T^{\pm}f(x)=f(x\pm1),
\quad
Rf(x)=f(-x),
\label{BIdunkl}
\end{equation}
\begin{equation}
F(x)=\frac{(x-\rho_1)(x-\rho_2)}{2x},
\qquad
G(x)=\frac{(x-r_1+\frac12)(x-r_2+\frac12)}{2x+1},
\end{equation}
with eigenvalues
\begin{equation}
\lambda_n=
\begin{cases}
\dfrac{n}{2}, & n \text{ even},\\[2ex]
r_1+r_2-\rho_1-\rho_2-\dfrac{n+1}{2}, & n \text{ odd}.
\end{cases}
\end{equation}
Setting $X=2L+\kappa$ with
$\kappa=\rho_1+\rho_2-r_1-r_2+\frac12$, $Y=2x+\frac12$ and $Z=\{X,Y\}-\omega_3$,
these operators realize the Bannai--Ito algebra \cite{BI}
\begin{equation}
\{X,Y\}=Z+\omega_3,
\qquad
\{Z,Y\}=X+\omega_1,
\qquad
\{X,Z\}=Y+\omega_2,
\label{BIalgebra}
\end{equation}
\begin{equation}
\omega_1=4(\rho_1\rho_2+r_1r_2),
\qquad
\omega_2=2\left(\rho_1^2+\rho_2^2-r_1^2-r_2^2\right),
\qquad
\omega_3=4(\rho_1\rho_2-r_1r_2),
\end{equation}
whose Casimir element is $Q=X^2+Y^2+Z^2$.

The complementary Bannai--Ito polynomials $W_n(x;\rho_1,\rho_2,r_1,r_2)$ are
the kernel polynomials of the $B_n$ with kernel parameter $\rho_1$
\cite{CBI}. The recurrence relation is
\begin{equation}
W_{n+1}(x)
+
(-1)^n \rho_2\, W_n(x)
+
\tau_n W_{n-1}(x)
=
x\,W_n(x),
\label{CBIrec}
\end{equation}
where the coefficients are given by
\begin{equation}
\tau_{2n}
=
-\frac{
n\left(n+\rho_1-r_1+\tfrac12\right)
\left(n+\rho_1-r_2+\tfrac12\right)
\left(n-r_1-r_2\right)
}{
(2n+g)(2n+g+1)
},
\label{CBItaueven}
\end{equation}
\begin{equation}
\tau_{2n+1}
=
-\frac{
(n+g+1)
\left(n+\rho_1+\rho_2+1\right)
\left(n+\rho_2-r_1+\tfrac12\right)
\left(n+\rho_2-r_2+\tfrac12\right)
}{
(2n+g+1)(2n+g+2)
},
\label{CBItauodd}
\end{equation}
and $g$ is given by \eqref{PBItruncation}. It is seen from \eqref{CBItaueven}--\eqref{CBItauodd} that the condition $\tau_n>0$
cannot be ensured for all $n$ and hence the CBI polynomials can only form a finite system of
positive-definite orthogonal polynomials
$W_0(x),\ldots,W_N(x)$, provided that the "local" positivity
$\tau_n>0$, $n\in\{1,\ldots,N\}$, and truncation conditions
$\tau_0=0$ and $\tau_{N+1}=0$ are satisfied.

When $N$ is even, the truncation conditions $\tau_0=0$ and
$\tau_{N+1}=0$ are equivalent to one of the four prescriptions
\begin{equation}
\begin{aligned}
1)\quad &\rho_2-r_1=-\frac{N+1}{2},
&\qquad
2)\quad &\rho_2-r_2=-\frac{N+1}{2},\\[6pt]
3)\quad &\rho_1+\rho_2=-\frac{N+2}{2},
&\qquad
4)\quad &g=-\frac{N+2}{2}.
\end{aligned}
\end{equation}
When $N$ is odd, the truncation conditions
$\tau_0=0$ and $\tau_{N+1}=0$ are tantamount to
\begin{equation}
\begin{aligned}
\text{i)}\quad &r_1-\rho_1=\frac{N+2}{2},\\[6pt]
\text{ii)}\quad &r_1+r_2=\frac{N+1}{2},\\[6pt]
\text{iii)}\quad &r_2-\rho_1=\frac{N+2}{2}.
\end{aligned}
\end{equation}
If these conditions are met, it leads to the standard truncated CBI polynomials.
Their algebra is \eqref{CBIalgebra}.

On the Bannai--Ito grid, the operators appearing in \eqref{BIdunkl} preserve
the set of spectral points and act on neighbouring indices: $R$ carries
$x_{2s}$ to $x_{2s-1}$ and $T^{+}R$ carries $x_{2s}$ to $x_{2s+1}$
\cite{BI,CBI}.  This is the property used in Subsection \ref{involution}.

\bibliographystyle{unsrt}
\bibliography{parafamilies}

\end{document}